\documentclass[11pt]{article}
\usepackage[utf8]{inputenc}
\usepackage[backend=bibtex,giveninits=true]{biblatex}
\usepackage{biblatex}
\usepackage{subfig}
\usepackage[normalem]{ulem}
\usepackage{float}

\bibliography{cc}

\usepackage{setspace}
\usepackage{xcolor}

\usepackage[margin=1in]{geometry}
\usepackage{comment}

\usepackage{mathtools}

\usepackage{amsmath,amsthm,amssymb,latexsym,amsfonts, dsfont}
\usepackage{graphicx} 
\graphicspath{ {\string~/Desktop/} }
\usepackage[font=footnotesize,labelfont=bf]{caption}
\usepackage{mathabx}
\usepackage{hyperref}
\usepackage{lineno}
\usepackage{amsthm}
\usepackage{bm}
\usepackage{listings}
\usepackage[nottoc,numbib]{tocbibind}
\usepackage{xcolor,colortbl}

\newtheorem{theorem}{Theorem}
\newtheorem{algorithm}{Algorithm}
\newtheorem{case}{Case}
\newtheorem{claim}{Claim}

\newtheorem{corollary}{Corollary}

\newtheorem{definition}{Definition}

\newtheorem{lemma}{Lemma}

\newtheorem{proposition}{Proposition}

\newtheorem{fact}{Fact}

\newcommand{\dn}{correlation metric}
\newcommand{\pr}{\mathbf{P}}
\newcommand{\midasterik}{\textasteriskcentered}
\newcommand{\eps}{\varepsilon}

\DeclareMathOperator*{\argmax}{arg\,max}
\DeclareMathOperator{\profit}{\text{profit}}
\allowdisplaybreaks

\DeclareCiteCommand{\citeyearpar}
    {}
    {\mkbibparens{\bibhyperref{\printdate}}}
    {\multicitedelim}
    {}

\title{Fast Combinatorial Algorithms for Min Max Correlation Clustering}

\author{Sami Davies\thanks{Northwestern University. Supported by an NSF Computing Innovation Fellowship. } \and Benjamin Moseley\thanks{Carnegie Mellon University. Benjamin Moseley and Heather Newman were supported in part by  a Google Research Award, an Inform Research Award, a Carnegie Bosch Junior Faculty Chair, and NSF grants CCF-2121744 and  CCF-1845146.} \and Heather Newman$^\dagger$}

\begin{document}
\maketitle

\abstract{
We introduce fast algorithms for correlation clustering with respect to the Min Max objective that provide constant factor approximations on complete graphs.    
Our algorithms are the \emph{first} purely combinatorial approximation algorithms for this problem.
We construct a novel semi-metric on the set of vertices, which we call the \dn, that indicates to our clustering algorithms whether pairs of nodes should be in the same cluster.
The paper demonstrates empirically that, compared to prior work, 
our algorithms sacrifice little in the objective quality
to obtain significantly better run-time. Moreover, our  algorithms scale to larger networks that are effectively intractable for known algorithms.
}

\section{Introduction}

In correlation clustering, a  graph $G= (V,E)$ is given as input, where each edge is labeled either positive $(+)$ or negative $(-)$. If two vertices are connected by a positive edge, this indicates they are similar and want to being clustered together. Alternatively, a negative edge indicates vertices that are dissimilar and want to be in different clusters. 
We make the common assumption that $G$ is complete; that is, there is a labelled edge between each pair of vertices. 

An edge is in \textit{disagreement} with respect to a clustering 
if it is a negative edge and connects two vertices contained inside the same cluster, or if it is a positive edge and connects two vertices in different clusters.
Note that among three or more vertices, 
the positive and negative labels may be such that \textit{any} clustering has some edges in disagreement, e.g. three vertices where the three labels between them are two positives and one negative.

The goal in correlation clustering is to find a clustering that minimizes an objective capturing the edges' disagreements.   
The most widely considered objective is to minimize the $\ell_p$-norm of the disagreements over the vertices. Here, each vertex $u$ has some number, $y(u)$, of disagreements adjacent to it and the goal is to minimize $\sqrt[p]{\sum_{u \in V} y(u)^p}$.  
The most well-studied objective is when $p = 1$, which minimizes the total number of disagreements; the Pivot algorithm is a popular, combinatorial algorithm in this setting \cite{ACN-pivot}.
The case where $p = \infty$ minimizes the maximum number of disagreements at any vertex, which captures a nice notion of fairness in the clustering.

Our results focus on the $\ell_\infty$-norm objective, which we will refer to as the \textit{Min Max} objective\footnote{Other names for this are the minimax or $\ell_\infty$ objective.}.  
Min Max correlation clustering was first motivated by applications in community detection that are antagonist free,  i.e., there are no members that are largely inconsistent in their community \cite{PM16}. Such problems arise in recommender systems, bioinformatics, and social sciences \cite{cheng2000biclustering, kriegel2009clustering, symeonidis2008nearest, PM16}.

Milenkovic and Puleo \citeyearpar{PM16} initiated the study of Min Max correlation clustering, as well as other $\ell_p$-norms. For all $p \geq 1$, they give a $48$-approximation. Charikar, Gupta, and Schwartz \citeyearpar{CGS17} improved this to a $7$-approximation, and Kalhan, Makarychev, and Zhou \citeyearpar{KMZ19} further improved this to the best known approximation\footnote{Note all of these results are still only for complete graphs.} of 5. 
Milenkovic and Puleo \citeyearpar{PM16} observed that Pivot fails for the $\ell_{\infty}$ norm, even though it gives a 3-approximation in expectation for the $\ell_1$ norm (see the discussion in Appendix A).
See more on correlation clustering in Section~\ref{sec:related}, 
including the interest on fair/local objectives like Min Max.

Previous algorithms on Min Max correlation clustering rounded SDP or LP solutions 
\cite{PM16, CGS17, Khuller2019min, KMZ19}.
The relaxations are large, requiring at least $|V|^3$ constraints and $|V|^2$ variables.  Thus solving them is not practical on even modest sized graphs (e.g. 300-500 vertices). The question looms: does there exist a fast algorithm for Min Max correlation clustering with strong theoretical guarantees?  Moreover, an intriguing direction both towards this question and in its own right is to develop combinatorial algorithms. A key challenge is that it is not clear how to compare to the optimal solution without the LP or SDP, as there are no
non-obvious lower bounds known.

\subsection{Our contributions}

The paper provides fast algorithms for Min Max correlation clustering. 
Our algorithms are the first combinatorial algorithms for the problem with theoretical guarantees.  Our key technical insight that enables these algorithms is the introduction of a semi-metric on the set of vertices, which we call the \textit{\dn}. 
The \dn~can be used to set variables in the problem's linear program so that the resulting solution is feasible and provably close to the optimal number of disagreements for the Min Max objective. Thus, we use the LP to compare to the optimal like prior work, but we do not need to solve it since we only use it in the analysis. Our \dn~gives new insights into both correlation clustering and the linear program.

Let $\omega$ denote the constant for 
the $O(n^{\omega})$ run-time of matrix multiplication on an $n \times n$ matrix, 
where $n = |V|$.  

\begin{theorem} \label{thm: apx_thm}
There is an algorithm that obtains a 40-approximation for Min Max
correlation clustering on complete graphs in time $O(n^{\omega})$.
\end{theorem}

The run-time can be substantially improved when the positive edges of the graph form a sparse graph, which is commonly the case in practice. 

\begin{corollary}\label{cor: apx-sparse}
Suppose all vertices in $V$ have $+$ degree at most $\Delta$,
and for each vertex we are given a list of its positive neighbors.
There is an algorithm that obtains a 40-approximation for 
Min Max correlation clustering on complete graphs in time $O(n \Delta^2 \log n)$.
\end{corollary}
Thus, we have a near-linear time algorithm for sparse graphs.  
This substantially improves on prior work, 
which (to the best of our knowledge) has run-time no better than $O(n^{2 \omega})$ even on sparse graphs
(see Appendix \ref{sec: lp_rounding_alg}).

Next, we show that we can approximate the \dn~to
improve the run-time, even when the graph is dense, 
though at some loss in the approximation factor.

\begin{theorem} \label{thm: apx_sampling}
Fix $\eps > 0$ sufficiently small.
For some constant $c(\eps)$ depending only on $\eps$, there is a randomized algorithm that obtains a $c(\eps)$-approximation for Min Max correlation clustering on complete graphs with probability $1-O(1/n)$ in time $O(n^2\log n/\eps^2)$. 
\end{theorem}

Finally, we show that our theory is predictive of practice (see Section \ref{sec: experiments}). Our algorithm performs similarly in terms of objective value to the algorithm of Kalhan, Makarychev, and Zhou \citeyearpar{KMZ19}, while improving the runtime so substantially that it can feasibly scale to graphs with about 10,000 vertices\footnote{For comparison, the algorithm by Kalhan, Makarychev, and Zhou takes more than 10 minutes on graphs with $\approx$ 200 vertices.}. Moreover, the clusters found by our algorithm are often meaningful in that they partially discover ``ground truth" clusters in real-world and synthetic instances.

\subsection{Related work}
\label{sec:related}

In correlation clustering, the most commonly studied objective is minimizing the
total number of edges in disagreement,
which is equivalent to minimizing the $\ell_1$ norm of the disagreement vector 
$y \in \mathbb{R}_{\geq 0}^{|E|}$.
Note that minimizing the $\ell_1$ norm can be studied when the edges are weighted.
The model was introduced by Bansal, Blum, and Chawla \citeyearpar{BBC04} and
is motivated by many applications such as image segmentation, 
natural language processing, clustering gene expression patterns, and location area planning \cite{Wirth17, mccallum2004conditional, ben1999clustering, demaine2003correlation}.
The celebrated Pivot algorithm of Ailon, Charikar, and Newman \citeyearpar{ACN-pivot}
obtains a 3-approximation in expectation. 
Until recently, the best approximation algorithm for $\ell_1$
correlation clustering on complete, unweighted graphs was a $(2.06+\epsilon)$-approximation due to Chawla, Makarychev, Schramm, and Yaroslavtsev
 \citeyearpar{chawla2015near}, 
but the threshold of 2 was broken 
(they achieved a ($1.994+\epsilon$)-approximation)
in the work of Cohen-Addad, Lee, and Newman \citeyearpar{cohen2022correlation},
who successfully took advantage of $O(1/\epsilon^2)$ rounds of the Sherali-Adams hierarchy.

Min Max correlation clustering was introduced 
by Puleo and Milenkovic \citeyearpar{PM16}.
They show correlation clustering for the Min Max objective is
\textsf{NP}-hard even on complete graphs (see Appendix C in their paper), and algorithmically, they obtain 48-approximation 
algorithms for complete graphs and complete bipartite graphs.
Shortly after, Charikar, Gupta, and Schwartz \citeyearpar{CGS17} gave a 7-approximation 
for minimizing the $\ell_p$ norm on the same graphs. 
For the Min Max objective, i.e. $p = \infty,$
they obtain a $O(\sqrt{n})$-approximation algorithm
on general, weighted graphs.
Kalhan, Makarychev, and Zhou \citeyearpar{KMZ19} further generalize the known results for 
correlation clustering with local objectives by showing approximation algorithms
that minimize the $\ell_p$ norm on general, weighted graphs.
This work also shows a 5-approximation 
for the $\ell_{p}$ (including $p = \infty$) objective on complete graphs and on complete bipartite graphs.

In all of the previous papers that study any $\ell_p$ norm objective 
for correlation clustering, 
the run-time relies on solving an LP with at least $\Omega(n^2)$ many variables and $\Omega(n^3)$ constraints \cite{PM16, CGS17, KMZ19}.
We see no clear way to use the structure of the LP to guarantee solving 
it would take time less than $O(n^{2 \omega})$, even on sparse graphs 
(see the discussion in Appendix \ref{sec: lp_rounding_alg}).

Several other objectives that are local or capture 
some notion of fairness have also been studied \cite{ahmadian2020fair, bateni2022scalable, ahmadi2020fair, friggstad2021fair, jafarov2021local, Khuller2019min}.
For instance, Ahmadi, Khuller, and Saha \citeyearpar{Khuller2019min} seek 
to find a clustering where for every point, 
the average distance to the points in its own cluster is
no more than the average distance to those in another cluster.
Their algorithm rounds the solution to an SDP, 
which was an idea based on work on min-max $k$ balanced partitioning.
Further, Jafarov et al. \citeyearpar{jafarov2021local} studied the $\ell_p$ objective, 
but additionally make assumptions on the weights based on whether 
edges are similar.

\section{Notation and Preliminaries}
Let $G=(V,E)$ be the complete graph on $n$ vertices with 
self-loops\footnote{One could avoid the self-loops by slightly changing the definition of neighborhood intersections that we use later. 
However, the notation and presentation are cleaner in just assuming positive self-loops.},
where each edge has a positive $(+)$ or negative $(-)$ label. 
Let $E^+$ be the set of positive edges of $G$ and let $E^-$ be the set of negative edges. 
For vertex $u \in V$, define $N_u^- = \{v \in V \mid (u,v) \in E^- \}$
to be the \textit{negative neighborhood} of $u$, and 
$N_u^+ = \{v \in V \mid (u,v) \in E^+ \}$
to be the \textit{positive neighborhood} of $u$. 

We say an edge $e=(u,v) \in E$ is in \textit{disagreement} according to a clustering if $e \in E^+$ and $u$ and $v$ are in different clusters or 
if $e \in E^-$ and $u$ and $v$ are in the same cluster. 
For ease in calculations, we assume that every vertex has a positive self-loop, i.e. $(v,v) \in E^+$ for all $v \in V$.
Note that doing so will not change the set of disagreeing edges;
since a vertex is always in the same cluster as itself, self-loops can never be disagreements. 
The following facts hold for $G$:

\begin{fact} \label{sum_of_nbhds}
For any $u \in V$, 
\[n= |N_u^+| + |N_u^-|. \]
\end{fact}

\begin{fact} \label{pairwise_intersections}
Let $u,v \in V$, possibly with $u=v$. Then
\[n = |N_u^+ \cap N_v^+| + |N_u^- \cap N_v^-| + |N_u^+ \cap N_v^-| + |N_u^- \cap N_v^+|.\]
\end{fact}

\begin{proof}[Proof of Fact \ref{pairwise_intersections}]
Observe that  $N_u^+ = (N_u^+ \cap N_v^+) \cup (N_u^+ \cap N_v^-)$ and $N_u^- = (N_u^- \cap N_v^+) \cup (N_u^- \cap N_v^-)$. 
Now substituting into Fact \ref{sum_of_nbhds} gives Fact \ref{pairwise_intersections}.
\end{proof}

Let $C = (C_1, \dots, C_k)$ denote a partition of $V$ into $k$ clusters. We say $C(u) = C_i$ if $u$ is in cluster $C_i$. Let $\overline{C(u)}$ be the set of all vertices in a different cluster from $u$. For a given clustering $C$, let $y_C \in \mathbb{Z}_+^n$ be the \textit{disagreement vector} of $C$, indexed by $V$; for $u \in V$, the coordinate $y_C(u)$ is the number of edges incident to $u$ that are disagreements in $C$. We will drop the subscript $C$ when the clustering is clear from context.

\paragraph{The correlation metric} 
We introduce a novel semi-metric
based on the positive and negative neighborhoods of vertices.  
We will prove that the correlation metric, $d$, satisfies the triangle inequality, cementing that it is a semi-metric.  

\begin{definition}
For all $u,v \in V$, the distance between $u$ and $v$ with respect to the \textit{\dn}~is
\begin{align}
 \label{eq: equiv_def}
d_{uv} &= 1-\frac{|N_u^+ \cap N_v^+|}{n -|N_u^- \cap N_v^-|} \notag \\
&= 1 - \frac{|N_u^+ \cap N_v^+|}{|N_u^+ \cap N_v^+| + |N_u^+ \cap N_v^-| + |N_u^- \cap N_v^+|}
\end{align}
\end{definition}
where (\ref{eq: equiv_def}) follows from Fact \ref{pairwise_intersections}. Observe that $d_{uv}$ is well-defined as $u,v$ have positive self-loops, so $0 \leq d_{uv} \leq 1$. 

\smallskip

We give some intuition on the definition of $d_{uv}$.
See Figure \ref{fig: duv}.
Fix three distinct vertices $u,v,w \in V$.
We will examine how the clustering of $u$ and $v$ affects the disagreements incident to them.
First, note that if $w \in N_u^- \cap N_v^-$, 
as long as $w$ is assigned to a different cluster 
than both $u$ and $v$, the edges $(u,w),(v,w)$ are not in disagreement,
regardless of whether or not $u$ or $v$ are in the same cluster. 
I.e., $w$ should not impact whether $u$ and $v$ are in the same cluster. 
If we do not subtract off $|N_u^- \cap N_v^-|$ in the normalization, we will not correctly identify perfect clusterings, whereas here, $d_{uv} = \mathds{1}_{\{C(u) \neq C(v)\}}$ when $C$ is a perfect clustering.

\begin{figure}
    \centering
\includegraphics[width = 14cm]{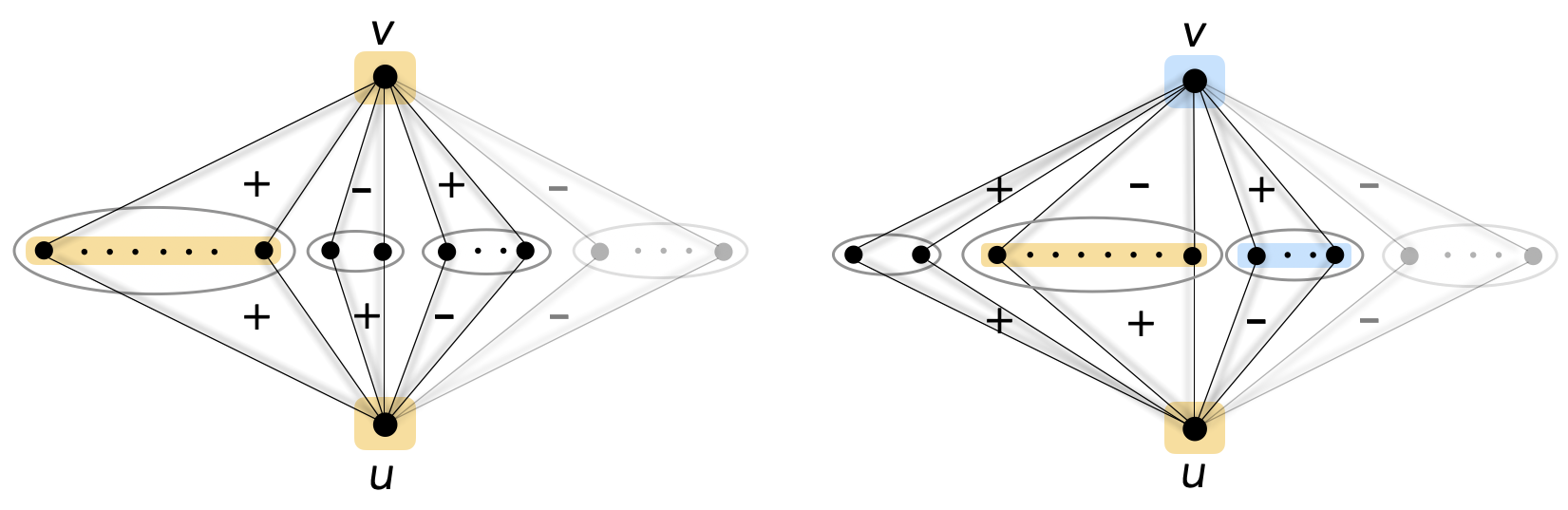}
    \caption{\textbf{Left:} $|N_u^+ \cap N_v^+|$ is significantly larger than $|(N_u^- \cap N_v^+) \cup (N_u^+ \cap N_v^-)|$. If $v$ and $u$ are in the same cluster, edges $(u,w)$ and $(w,v)$ do not form disagreements for all $w \in N_u^+ \cap N_v^+$.
    \textbf{Right:} $|(N_u^- \cap N_v^+) \cup (N_u^+ \cap N_v^-)|$ is significantly larger than $|N_u^+ \cap N_v^+|$. If $v$ and $u$ are in different clusters (yellow and blue), edges $(u,w)$ and $(w,v)$ do not form disagreements for all $w \in (N_u^- \cap N_v^+) \cup (N_u^+ \cap N_v^-)$.}
    \label{fig: duv}
\end{figure}

On the other hand, if $w \not \in N_u^- \cap N_v^-$, 
then whether or not $u$ and $v$ are in the same cluster
directly impacts whether either of $(u,w)$ or $(v,w)$ are in disagreement. 
By the proof of Fact \ref{pairwise_intersections}, the edges that are affected are $(u,w)$
and $(v,w)$ for
$w \in (N_u^+ \cap N_v^+) \cup (N_u^- \cap N_v^+) \cup (N_u^+ \cap N_v^-)$.
If $u$ and $v$ are in different clusters,
then for $w \in (N_u^- \cap N_v^+) \cup (N_u^+ \cap N_v^-)$,
both $(u,w)$ and $(v,w)$ may still be in agreement. 
However, if $w \in N_u^+ \cap N_v^+$, 
then at least one of $(u,w)$ and $(v,w)$ 
are in disagreement, 
so the number of disagreements incident to $u$ and $v$ (in total)
is at least $|N_u^+ \cap N_v^+|$.
An analogous point can be seen when $u$ and $v$ are in the same cluster, 
but now there are at least $|(N_u^- \cap N_v^+) \cup (N_u^+ \cap N_v^-)|$ disagreements 
incident to $u$ and $v$.
Overall, if $|N_u^+ \cap N_v^+|$ is large relative to 
$|(N_u^- \cap N_v^+) \cup (N_u^+ \cap N_v^-)|$ then (roughly speaking) there are fewer disagreements 
incident to $u$ and $v$ if they are assigned to the same cluster than if they were assigned to different clusters. 
The larger $|N_u^+ \cap N_v^+|$ is relative to 
$|(N_u^- \cap N_v^+) \cup (N_u^+ \cap N_v^-)|$, the closer $u$ and $v$ are with respect to the \dn, by (\ref{eq: equiv_def}).

Throughout the paper define $\widehat{d}_{uv} = d_{uv}$ if $(u,v) \in E^+$ and 
$\widehat{d}_{uv} = 1 - d_{uv}$ if $(u,v) \in E^-$.

\section{Technical Overviews}\label{sec: overviews}

\subsection{LP relaxation and KMZ algorithm}

We consider the standard LP relaxation for the problem.
It is not hard to see that the LP's constraints induce a semi-metric space on the vertices.  
In the algorithm of Kalhan, Makarychev, and Zhou, the LP is solved to get a fractional solution, and then an iterative rounding  procedure is used to get an integral feasible clustering. 
Intuitively, the LP semi-metric guides the solution, where vertices close together are more likely to end up in the same cluster. 
In fact, we will never place vertices sufficiently far apart in the same cluster.

The LP consists of variables $x_{uv}$.  In an integral solution, $x_{uv}=1$ indicates $u$ and $v$ are in different clusters and $x_{uv}=0$ indicates they are in the same cluster. The disagreement vector is $y$. 

\vspace{-1.2cm}
\begin{flushleft}
\begin{equation*}\label{KMZ_LP}
\end{equation*}
LP \ref{KMZ_LP}
\end{flushleft}
\vspace{-1cm}
\begin{align}
     &\min \max_{u \in V} y(u) \notag\\
    \textsf{s.t. } y_u &= \sum_{v \in N_u^+} x_{uv} + \sum_{v \in N_u^-} (1-x_{uv})  && \forall u \in V \notag\\
     x_{uv}&\leq x_{uw}+x_{vw}  &&\forall u,v,w \in V \notag\\
     0 &\leq x_{uv} \leq 1  &&\forall u,v \in V. \notag
\end{align}

The first set of constraints ensure that the disagreement vector $y$ is set according to the $x$ variables. The second constraints enforce the triangle inequality on $x$. 
They ensure that if $u$ and $w$ are in the same cluster and $v$ and $w$ are in the same cluster then $u$ and $v$ must be together.

We will refer to the algorithm by Kalhan, Makarychev, and Zhou 
as the KMZ algorithm. 
The \textit{KMZ algorithm} has two phases: 
first it solves LP \ref{KMZ_LP} above, 
and then it rounds the LP solution in the \textit{rounding algorithm} (Algorithm \ref{KMZ-alg}). 
We justify in Appendix \ref{sec: lp_rounding_alg} that 
the run-time of the rounding algorithm is $O(n^2)$, 
so the run-time of the KMZ algorithm is dominated by the time 
it takes to solve the LP, 
which to the best of our knowledge has run-time no better than $O(n^{2 \omega})$.

\subsection{Combinatorial algorithms}

We begin with an outline for the proof of Theorem \ref{thm: apx_thm}; 
details can be found in Section \ref{sec: apx-thm-pf}.
Our first algorithm consists of two steps.
First, we compute the correlation metric $d_{uv}$ for all $u,v \in V$, which produces a hand-crafted solution to LP \ref{KMZ_LP} by setting $x_{uv}=d_{uv}$.
Then, we use the $d_{uv}$ values as input to the rounding algorithm (Algorithm \ref{KMZ-alg}).
Our analysis to prove that this procedure gives a constant factor approximation has three key steps: 
 \begin{enumerate}
     \item The \dn~(perhaps surprisingly) satisfies the triangle inequality (see Section~\ref{sec: triangle_ineq}), i.e. $x=d$ is feasible for LP \ref{KMZ_LP}.
     \item The objective to LP \ref{KMZ_LP} from setting $x=d$ is at most a factor 8 more than an optimal integral solution (see Section~\ref{sec: bound_frac_cost}). This is the heart of Theorem \ref{thm: apx_thm}, and is tricky due to asymmetries in the fractional cost that force us to use non-local charging arguments. 
     \item The rounding algorithm (Algorithm \ref{KMZ-alg}) of Kalhan, Makarychev, and Zhou  \citeyearpar{KMZ19} can be used to round any feasible solution to LP \ref{KMZ_LP} to an integer solution, while losing a factor  of at most $5$ in the objective. 
 \end{enumerate}

Our run-time is dominated by the time to compute $d_{uv}$
for all $u,v \in V$, which can be done in time $O(n^{\omega})$.
For run-time proofs, see Section \ref{sec: complete-proof} 
and Appendix \ref{sec: lp_rounding_alg}.

For sparse graphs,
instead of computing the correlation metric
for all pairs $u,v \in V$,
we need only compute $d_{uv}$ when $d_{uv} < 1$ (there are at most $n\Delta^2$ such pairs). Other pairs are implicitly distance $1$.
Then we can again use the rounding algorithm (Algorithm \ref{KMZ-alg}).  
Given a feasible solution to LP \ref{KMZ_LP}, the rounding algorithm runs in time $O(n \Delta^2 \log n)$ for sparse graphs.
See Section \ref{sec: complete-proof} for the complete proof for Corollary \ref{cor: apx-sparse}.

Lastly, in Section \ref{sec: sampling} we show that we can estimate the $d_{uv}$ via sampling, giving an algorithm that trades the guarantee on the (still constant) approximation factor for a faster run-time, leading to Theorem \ref{thm: apx_sampling}. To compute the estimates, we sample $O(\log n)$ vertices from the positive neighborhood of each vertex. Then, we use the samples for $u$ and $v$ to obtain estimates of the various terms in (\ref{eq: equiv_def}). However, it is not clear a priori that sampling will work. First, not every term in (\ref{eq: equiv_def}) will be well-concentrated, so we will need to exploit the special structure of $d_{uv}$. Second, the fractional cost of the initial estimates will not be controlled. We will show a post-processing phase, in which we push some estimates down and others up, takes care of the issue. We use the post-processed estimates as input to the rounding algorithm.

\section{Properties of the Correlation Metric }\label{sec: apx-thm-pf}
This section is focused on proving Theorem \ref{thm: apx_thm}.  
It establishes structural properties of the \dn.  
We refer to $\max_{u \in V} \{ \sum_{v \in N^+_u} d_{uv} + \sum_{v \in N^-_u} (1-d_{uv})\}$ as the fractional cost, which is the objective value of the LP variables we set.  The fractional cost is bounded in Section \ref{sec: bound_frac_cost}.
Section \ref{sec: triangle_ineq} establishes that the \dn~is a semi-metric, and therefore is feasible for LP \ref{KMZ_LP}. 
Finally, we derive Theorem \ref{thm: apx_thm} in Section \ref{sec: complete-proof}.

\subsection{Bounding the fractional cost}\label{sec: bound_frac_cost}

We begin with a few propositions that will be helpful for bounding the fractional cost.

\begin{proposition} \label{disagreements_mixed}
For any vertex $u \in V$, let $v \in N_u^+ \cap C(u)$. Then 
\[|N_u^+ \cap N_v^-| + |N_u^- \cap N_v^+| \leq y(u) + y(v). \]
\end{proposition}

\begin{proof}[Proof of Proposition \ref{disagreements_mixed}]
Take $S = (N_u^+ \cap N_v^-) \cup (N_u^- \cap N_v^+)$.
Note that $N_u^+ \cap N_v^-$ and $N_u^- \cap N_v^+$ are disjoint,
so $|S| = |N_u^+ \cap N_v^-| + |N_u^- \cap N_v^+|$.
Fix $w \in S$. 
Since $u$ and $v$ are in the same cluster, 
exactly one of $(u,w)$ and $(v,w)$ is a disagreement, 
and thus contributes to one of $y(u)$ or $y(v)$.
\end{proof}

Our main technical result to prove Theorem \ref{thm: apx_thm} is Lemma \ref{lem: main_lem}, though to make the proof of the lemma
a bit more modular, 
we rely on Claims \ref{sum_plus} and \ref{sum_minus}.
Recall $\widehat{d}_{uv} = d_{uv}$ if $(u,v) \in E^+$ and 
$\widehat{d}_{uv} = 1 - d_{uv}$ if $(u,v) \in E^-$. 

\begin{lemma} \label{lem: main_lem}
Let $y$ be the disagreement vector for an optimal solution to the min-max objective. For every $u \in V$ and for \textsf{OPT}$=\max_{z \in V} y(z)$, 
\[\sum_{v \in V} \widehat{d}_{uv} \leq 8 \cdot  \textsf{OPT}. \]
\end{lemma}

\begin{proof}[Proof of Lemma \ref{lem: main_lem}]
Expanding the summation, we have:
\begin{align*}
\sum_{v \in V} \widehat{d}_{uv} &= \sum_{v \in N_u^+} d_{uv} + \sum_{v \in N_u^-} (1-d_{uv}) \\
&= \underbrace{\sum_{v \in N_u^+} \frac{|N_u^+ \cap N_v^-| + |N_u^- \cap N_v^+|}{n-|N_u^- \cap N_v^-|}}_{S^+} + \underbrace{\sum_{v \in N_u^-} \frac{|N_u^+ \cap N_v^+|}{n-|N_u^- \cap N_v^-|}}_{S^-}
\end{align*}
We use Fact \ref{pairwise_intersections} to rewrite the first summation. Let the first summation be denoted $S^+$ and the second summation $S^-$. To prove the lemma, it suffices to prove the following two claims.

\begin{claim} \label{sum_plus}
\[S^+ \leq 3 \cdot \textsf{OPT}. \]
\end{claim}

\begin{claim} \label{sum_minus}
\[S^- \leq 5 \cdot \textsf{OPT}. \]
\end{claim}
\end{proof}

The proof of Claim \ref{sum_plus} is much cleaner than the proof 
of Claim \ref{sum_minus}. 
This is in part due to asymmetries in the fractional cost. 
\begin{proof}[Proof of Claim \ref{sum_plus}]
Observe that by Fact \ref{pairwise_intersections}, for every $v \in V$, 
\[|N_u^+| = |N_u^+ \cap N_v^-| + |N_u^+ \cap N_v^+|  \leq n-|N_u^- \cap N_v^-|. \]
Now we can bound $S^+$ by partitioning the sum based 
one whether or not $u$ and $v$ are in the same cluster:
\begin{align}
    S^+ &= \sum_{v \in N_u^+ \cap C(u)} \frac{|N_u^+ \cap N_v^-| + |N_u^- \cap N_v^+|}{n-|N_u^- \cap N_v^-|} + \sum_{v \in N_u^+ \cap \overline{C(u)}} d_{uv} \notag \\
    &\leq \frac{1}{|N_u^+|} \cdot \sum_{v \in N_u^+ \cap C(u)}(|N_u^+ \cap N_v^-| + |N_u^- \cap N_v^+|) + \sum_{v \in N_u^+ \cap \overline{C(u)}} 1 \notag \\
    &\leq \frac{1}{|N_u^+|} \cdot \sum_{v \in N_u^+ \cap C(u)} (y(u) + y(v)) + \sum_{v \in N_u^+ \cap \overline{C(u)}} 1 \label{prop2} \\
    &\leq y(u) + \max_z y(z) + \sum_{v \in N_u^+ \cap \overline{C(u)}} 1 \label{avging_arg} \\
    &\leq  y(u) + \max_z y(z) + y(u) \label{easy_pos} \\
    &\leq 3 \cdot \textsf{OPT}. \notag
\end{align}
Line (\ref{prop2}) follows from Proposition \ref{disagreements_mixed}, line (\ref{avging_arg}) follows from an averaging argument (we sum over $|N_u^+ \cap C(u)|$ terms and then divide by $|N_u^+|$), and line (\ref{easy_pos}) follows from the fact that if $(u,v) \in E^+$ and $u,v$ are in different clusters, then $(u,v)$ is a disagreement incident to $u$. 
\end{proof}

\begin{proof}[Proof of Claim \ref{sum_minus}]
We have 
\begin{align*}
    S^- &= \sum_{v \in N_u^- \cap C(u)} (1-d_{uv}) + \sum_{v \in N_u^- \cap \overline{C(u)}} \frac{|N_u^+ \cap N_v^+|}{n-|N_u^- \cap N_v^-|} \\
    &\leq \sum_{v \in N_u^- \cap C(u)} 1 + \sum_{v \in N_u^- \cap \overline{C(u)}}  \frac{|N_u^+ \cap N_v^+|}{n-|N_u^- \cap N_v^-|} \\
    &\leq y(u) + \underbrace{\sum_{v \in N_u^- \cap \overline{C(u)}}  \frac{|N_u^+ \cap N_v^+|}{n-|N_u^- \cap N_v^-|}}_{S},
\end{align*}
where the last inequality follows from the fact that if $(u,v) \in E^-$ and $u,v$ are in the same cluster, then $(u,v)$ is a disagreement incident to $u$. 

It remains to bound the last summation, call it $S$. While we might expect the argument to mimic the bounding of the sum $\sum_{v \in N_u^+ \cap C(u)} d_{uv}$ in the proof of Claim \ref{sum_plus}, e.g., via an averaging argument, this will not work. 
Overall, for $v\in N_u^+$ we can compare the sum $\sum_{v \in N_u^+} d_{uv}$ only to disagreements incident to $u$ and to $v$,
but for $v \in N_u^- \cap \overline{C(u)}$ we will have to compare to disagreements incident to vertices besides $u$ and $v$. 
Intuitively, this is because putting $u$ and $v$ in different clusters 
means $(u,v)$ is not in disagreement according to the clustering, 
but necessarily any vertex $w \in N_u^+ \cap N_v^+$ is incident to 
an edge in disagreement.
Further, this disagreement can be charged to \textit{another} vertex---a carefully chosen $v^*(w) \in C(w)$.

We begin by ``flipping" the sum $S$. In particular, we would like to view the sum as being taken over elements in $N_u^+ \cap N_v^+$, subject to scaling by $n-|N_u^- \cap N_v^-|$.

\begin{align*}
    S &= \sum_{v \in N_u^- \cap \overline{C(u)}} \frac{|N_u^+ \cap N_v^+|}{n-|N_u^- \cap N_v^-|} = \sum_{w \in N_u^+}\sum_{\substack{v: \hspace{0.1 cm} v \in N_w^+,\\ v \in N_u^- \cap \overline{C(u)}}} \frac{1}{n-|N_u^- \cap N_v^-|}.
\end{align*}

Define for each $w \in N_u^+$:
\[V(w) := \{v \in V \mid v \in N_u^- \cap N_w^+ \cap \overline{C(u)} \} \]

\begin{figure}[H]
\centering
\includegraphics[width=0.6\textwidth]{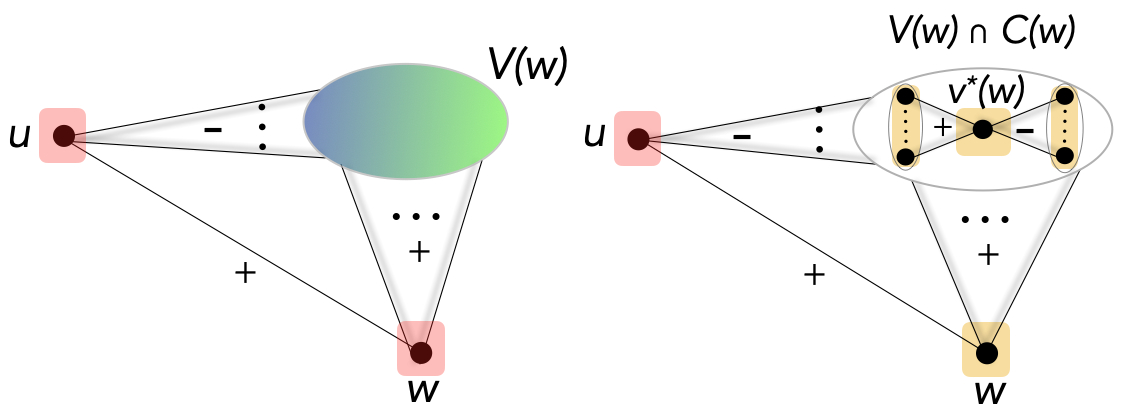}
\captionsetup{width=.8\linewidth}
\caption{\textbf{Left:} A representation of bounding the sum $S_1$. Here, $w$ is in the same cluster as $u$ (red), while the vertices in $V(w)$ are in various other clusters (blue/green). \textbf{Right:} A representation of bounding the sum $S_2(w)$, specifically the right-hand side in line (\ref{y(w)_again}). Yellow vertices are in the same cluster. The vertex $v^*(w)$ is carefully chosen as in the text.}
\label{fig: bounding_diagram}
\end{figure}

Note that $V(u) = \emptyset$, since $N_u^- \cap N_u^+ = \emptyset$, so the outer sum need not include the self-loop from $u$. Thus we have:
\begin{align*}
    S &= \sum_{w \in N_u^+ \setminus u} \hspace{0.1cm} \sum_{v \in V(w)} \frac{1}{n-|N_u^- \cap N_v^-|} \\
    &= \underbrace{\sum_{\substack{w \in N_u^+ \setminus u,\\ w \in C(u)}} \hspace{0.1cm} \sum_{v \in V(w)} \frac{1}{n-|N_u^- \cap N_v^-|}}_{S_1} + \underbrace{\sum_{\substack{w \in N_u^+,\\ w 
    \in \overline{C(u)}}} \hspace{0.1cm} \sum_{v \in V(w)} \frac{1}{n-|N_u^- \cap N_v^-|}}_{S_2}
\end{align*}
Let $S_1$ denote the left-hand sum and $S_2$ denote the right-hand sum. First we will bound $S_1$ (see Figure \ref{fig: bounding_diagram}):

\begin{align}
    S_1 &= \sum_{\substack{w \in N_u^+ \setminus u,\\ w \in C(u)}} \hspace{0.1cm} \sum_{v \in V(w)} \frac{1}{n-|N_u^- \cap N_v^-|} 
    \leq \sum_{\substack{w \in N_u^+ \setminus u,\\ w \in C(u)}} \hspace{0.1cm} \sum_{v \in V(w)} \frac{1}{|N_u^+|} \label{unif_bd}\\
    &= \sum_{\substack{w \in N_u^+ \setminus u,\\ w \in C(u)}} \frac{|V(w)|}{|N_u^+|}
    \leq \sum_{\substack{w \in N_u^+ \setminus u,\\ w \in C(u)}} \frac{y(w)}{|N_u^+|} \label{y(w)} \\
    &\leq \frac{\textsf{OPT}}{|N_u^+|} \cdot |(N_u^+ \setminus u) \cap C(u)| \notag \leq \textsf{OPT}. \notag
\end{align}

The inequality in Line (\ref{unif_bd}) follows from Fact \ref{sum_of_nbhds}
along with the fact that $|N_u^- \cap N_v^-| \leq |N_u^-|$. 
The inequality in Line (\ref{y(w)}) follows from 
the fact that if $w \in C(u) \setminus u$, 
then every member of $V(w)$ 1) is in a different cluster than $w$, 
and 2) has a positive edge to $w$; 
so $|V(w)|$ is at most the number of disagreements incident to $w$.

Bounding $S_2$ is more involved. For $w \in N_u^+ \cap \overline{C(u)}$, define 
\[S_2(w) := \sum_{v \in V(w)} \frac{1}{n-|N_u^- \cap N_v^-|}, \]
so that we can write 
\[S_2 = \sum_{\substack{w \in N_u^+,\\ w \in \overline{C(u)}}} S_2(w). \]
Fix $w \in N_u^+ \cap \overline{C(u)}$. We will bound $S_2(w)$ (see Figure \ref{fig: bounding_diagram}).  Note that neither $u$ nor $w$ is in $V(w)$. 

\begin{align}
    S_2(w) &=  \sum_{\substack{v \in V(w),\\ v \in \overline{C(w)}}} \frac{1}{n-|N_u^- \cap N_v^-|} +  \sum_{\substack{v \in V(w),\\ v \in C(w)}} \frac{1}{n-|N_u^- \cap N_v^-|}  \notag \\
    &\leq \frac{|V(w) \cap \overline{C(w)}|}{|N_u^+|} +  \sum_{\substack{v \in V(w),\\ v \in C(w)}} \frac{1}{n-|N_u^- \cap N_v^-|} \label{unif_bd_again}\\
    &\leq \frac{y(w)}{|N_u^+|} + \sum_{\substack{v \in V(w),\\ v \in C(w)}} \frac{1}{n-|N_u^- \cap N_v^-|}. \label{y(w)_again}
\end{align}
In line (\ref{unif_bd_again}) we have again used the bound that $n-|N_u^- \cap N_v^-| \geq |N_u^+|$, and in line (\ref{y(w)_again}) we have used that the edges from $w$ to $V(w) \cap \overline{C(w)}$ are positive and therefore in disagreement. 

Define 
\[v^*(w) = \argmax_{v \in V(w) \cap C(w)} |N_u^- \cap N_v^-| \]
so that $|N_u^- \cap N_v^-| \leq |N_u^- \cap N_{v^*(w)}^-|$ for all $v \in V(w) \cap C(w)$.
Continuing from line (\ref{y(w)_again}), we have: 

\begin{align}
    S_2(w) &\leq \frac{y(w)}{|N_u^+|} + \sum_{\substack{v \in V(w),\\ v \in C(w),\\ v \in N_{v^*(w)}^+}} \frac{1}{n-|N_u^- \cap N_v^-|} + \sum_{\substack{v \in V(w),\\ v \in C(w),\\ v \in N_{v^*(w)}^-}} \frac{1}{n-|N_u^- \cap N_v^-|} \label{v*_split} \\
    &\leq \frac{y(w)}{|N_u^+|} + \sum_{\substack{v \in V(w),\\ v \in C(w),\\ v \in N_{v^*(w)}^+}} \frac{1}{n-|N_u^- \cap N_{v^*(w)}^-|} + \sum_{\substack{v \in V(w),\\ v \in C(w),\\ v \in N_{v^*(w)}^-}} \frac{1}{|N_u^+|} \label{diff_bds}\\
    &\leq \frac{y(w)}{|N_u^+|} + \sum_{\substack{v \in V(w),\\ v \in C(w),\\ v \in N_{v^*(w)}^+}} \frac{1}{n-|N_u^- \cap N_{v^*(w)}^-|} +  \frac{y(v^*(w))}{|N_u^+|} \label{neg_disag}\\
    &= \frac{y(w)}{|N_u^+|} + \frac{|V(w) \cap C(w) \cap N_{v^*(w)}^+|}{n-|N_u^- \cap N_{v^*(w)}^-|} + \frac{y(v^*(w))}{|N_u^+|} \notag \\
    &\leq \frac{y(w)}{|N_u^+|} + \frac{|N_u^- \cap N_{v^*(w)}^+|}{n-|N_u^- \cap N_{v^*(w)}^-|} + \frac{y(v^*(w))}{|N_u^+|} \label{uv*_int}\\
    &\leq \frac{y(w)}{|N_u^+|} + 1 + \frac{y(v^*(w))}{|N_u^+|} \label{bd_by_1}\\
    &\leq  1 + \frac{2 \cdot \textsf{OPT}}{|N_u^+|} \notag
\end{align}

Line (\ref{v*_split}) follows from the fact that $v^*(w) \in N_{v^*(w)}^+$. Line (\ref{diff_bds}) follows from the definition of $v^*(w)$ along with the previously used bound of $n-|N_u^- \cap N_v^-| \geq |N_u^+|$. Line (\ref{neg_disag}) follows from the fact that $v^*(w)$ and $N_{v^*(w)}^- \cap C(w)$ are in the same cluster $C(w)$, so the edges between $v^*(w)$ and $N_{v^*(w)}^- \cap C(w)$ are disagreements incident to $v^*(w)$. Line (\ref{uv*_int}) follows from the fact that every vertex in $V(w)$ is in $N_u^-$, and line (\ref{bd_by_1}) follows from Fact \ref{pairwise_intersections}.\\

Finally, having a bound for $S_2(w)$, we can bound $S_2$ and then $S$ and $S^-$, which will finish the proof of Claim \ref{sum_minus}. 
\begin{align}
    S_2 &= \sum_{\substack{w \in N_u^+,\\ w \in \overline{C(u)}}} S_2(w) \notag 
    \leq \sum_{\substack{w \in N_u^+,\\ w \in \overline{C(u)}}} \left(\frac{2 \cdot \textsf{OPT}}{|N_u^+|} + 1\right) \notag 
    \leq 2 \cdot \textsf{OPT} + \sum_{\substack{w \in N_u^+,\\ w \in \overline{C(u)}}} 1 \notag \\
    &\leq 2 \cdot \textsf{OPT} + y(u) \label{pos_disag} \notag 
    \leq 3 \cdot \textsf{OPT}. \notag
\end{align}
Finally, we finish the proof of the claim because
\[S^- \leq \textsf{OPT} + S = \textsf{OPT} + S_1 + S_2 \leq \textsf{OPT} + \textsf{OPT} + 3 \cdot \textsf{OPT} = 5 \cdot \textsf{OPT}. \]
\end{proof}

We remark that in bounding $S$, the argument is non-local.  Rather, we used global properties of the graph.  We believe bounding the optimal value requires global structural properties of the graph.

\subsection{Triangle inequality}\label{sec: triangle_ineq}
We will show that the \dn~satisfies the triangle inequality, i.e., 
for distinct $u,v,w \in V$, 
$d_{uv} + d_{vw} \geq d_{uw}$. 

\begin{lemma}
\label{lem: triangle-in}
The \dn~satisfies the triangle inequality.
\end{lemma}
\begin{proof}
Fix three distinct vertices $u,v,w \in V$. We see from the definition of $d$ that 
$d_{uv} + d_{vw} \geq d_{uw}$ if and only if
\begin{align*}
1- \frac{|N_u^+ \cap N_v^+|}{n-|N_u^- \cap N_v^-|} + 1- \frac{|N_w^+ \cap N_v^+|}{n-|N_w^- \cap N_v^-|} &\geq 1- \frac{|N_u^+ \cap N_w^+|}{n-|N_u^- \cap N_w^-|},
\end{align*}
which after rearranging is equivalent to
\begin{align*}
\frac{|N_u^+ \cap N_w^+|}{n-|N_u^- \cap N_w^-|} +1 & \geq \frac{|N_u^+ \cap N_v^+|}{n-|N_u^- \cap N_v^-|} + \frac{|N_w^+ \cap N_v^+|}{n-|N_w^- \cap N_v^-|}.
\end{align*}
For shorthand, allow $a_{uv} = n-|N_u^- \cap N_v^-|$, and analogously define $a_{uw}$ and $a_{vw}$.
Then we can multiply both sides of the above inequality by $a_{uw}\cdot a_{uv}\cdot a_{vw}$ to rewrite it as
\begin{equation}
\label{eq: a-intro}
a_{uv}a_{wv} \left (|N_u^+\cap N_w^+|+  n-|N_u^- \cap N_w^-| \right )\geq a_{uw}a_{vw}|N_u^+ \cap N_v^+| + a_{uw}a_{uv}|N_w^+ \cap N_v^+|.
\end{equation}

We can rewrite the left-hand side of Equation \ref{eq: a-intro} using Fact \ref{pairwise_intersections} and expand it into the intersection of all three vertices as 
\begin{align*}
    a_{uv}a_{vw} \left (|N_u^+ \cap N_w^+|+a_{uw} \right )&=
    a_{uv}a_{vw}  (|N_u^+ \cap N_w^+|+ |N_u^+ \cap N_w^+|+ |N_u^- \cap N_w^+|+ |N_u^+ \cap N_w^-|  )\\
    &= a_{uv}a_{vw} (2|N_u^+ \cap N_v^+ \cap N_w^+|+2|N_u^+ \cap N_v^- \cap N_w^+|\\
    & \qquad + |N_u^- \cap N_v^+\cap N_w^+|
     + |N_u^- \cap N_v^-\cap N_w^+|\\
     & \qquad + |N_u^+ \cap N_v^+\cap N_w^-|+|N_u^+ \cap N_v^-\cap N_w^-|).
\end{align*}
We introduce more shorthand to compactly notate the intersections of 3 neighborhoods. Let the subscripts of the variable $b$ denote the vertices whose positive intersection we examine, so $b_{uvw} = |N_u^+ \cap N_v^+\cap N_w^+|$, $b_{uv} = |N_u^+ \cap N_v^+\cap N_w^-|$, $b_{u} = |N_u^+ \cap N_v^-\cap N_w^-|$, etc.
Therefore we can more compactly write the left-hand side of Equation \ref{eq: a-intro} as 
\begin{equation}
    \label{eq: b-intro}
    a_{uv}a_{vw} \left (|N_u^+ \cap N_w^+|+a_{uw} \right ) =a_{uv}a_{vw} \left (
     2b_{uvw}+2b_{uw}+ b_{uv}+b_{vw}+ b_{u}+b_w \right ).
\end{equation}
We next write the factors $a_{uv},a_{uw}$, and $a_{vw}$ in terms of the $b$ variables, where we let $B=b_{uvw}+ b_{uv}+b_{uw}+b_{vw}$:
\begin{align*}
    a_{uv} &= |N_u^+ \cap N_v^+|+|N_u^+ \cap N_v^-|+|N_u^- \cap N_v^+|
    =b_{uv}+b_{uvw}+b_u+b_{uw}+b_v+b_{vw} \\
    &= B+b_u+b_v
\end{align*}
Similarly, $a_{uw}=B+b_u+b_w$ and $a_{vw} =B+b_v+b_w.$
Plugging these into Equation \ref{eq: b-intro} the full left-hand side of Equation \ref{eq: a-intro} is
\begin{equation}
\label{eq: LHS-tri-in}
 a_{uv}a_{vw} \left (|N_u^+ \cap N_w^+|+a_{uw} \right ) = (B+b_u+b_v)(B+b_v+b_w) \left (
     B+b_{uvw}+b_{uw}+ b_{u}+b_w \right ).
\end{equation}

Now we move onto rewriting the right-hand side of Equation \ref{eq: a-intro}:
\begin{align}\label{eq: RHS}
    a_{uw}&a_{vw}|N_u^+ \cap N_v^+| + a_{uw}a_{uv}|N_w^+ \cap N_v^+| \notag \\
    &= (B+b_u+b_{w})(B+b_v+b_{w})(b_{uvw}+b_{uv})+(B+b_u+b_{w})(B+b_u+b_{v})(b_{uvw}+b_{vw}).
\end{align}
Using Equations \ref{eq: LHS-tri-in} and \ref{eq: RHS}, we rewrite the condition in Equation \ref{eq: a-intro} as 
\begin{align*}
 (B+b_u+b_v)&(B+b_v+b_w) \left (
     b_{uvw}+b_{uw}+B+b_u+ b_{w}\right ) \\
     & \geq (B+b_u+b_w)(B+b_v+b_w)(b_{uvw}+b_{uv})\\
    & \qquad +(B+b_u+b_w)(B+b_u+b_v)(b_{uvw}+b_{vw}),
\end{align*}
which one can verify is true as every term on the 
right-hand side is on the left-hand side too.
\end{proof}

\subsection{Completing the proofs of Theorem \ref{thm: apx_thm} and Corollary \ref{cor: apx-sparse}}\label{sec: complete-proof}

Tying it all together, we prove the first combinatorial approximation algorithm for Min Max correlation clustering. 

\begin{proof}[Proof of Theorem \ref{thm: apx_thm}]

For the clustering $\mathcal{C}$ output by the rounding algorithm 
(Algorithm \ref{KMZ-alg}) of Kalhan, Makarychev, and Zhou with the \dn~as input, 
let $\textsf{ALG}(u,v) = \mathds{1}((u,v) \text { is in }$ $\text{disagreement in }\mathcal{C})$ 
and $\textsf{ALG}(u) = \sum_{v \in V} \textsf{ALG}(u,v)$.
We see that
\begin{equation} \label{eq: KMZ_chain}
    \textsf{ALG}(u) = \sum_{v \in V}\textsf{ALG}(u,v)  \underset{*}{\leq} 
5 \cdot \sum_{v \in V}\widehat{d}_{uv} \leq 40 \cdot \textsf{OPT}.
\end{equation} 
The first inequality follows because the \dn~is a feasible solution to
LP \ref{KMZ_LP}, since all $0 \leq d_{uv} \leq 1$ 
and by Lemma \ref{lem: triangle-in} 
$d$ satisfies the triangle inequality.
The second inequality follows by Lemma \ref{lem: main_lem}.

Next, we analyze the run-time.
Our algorithm has two phases:
in phase (1) it computes the \dn~ for all $u,v \in V$, i.e. $d_{uv}$,
and in phase (2) it uses the rounding algorithm (Algorithm \ref{KMZ-alg}) with input $d$. 
Phase (1) of our algorithm takes $O(n^\omega)$ time, where $\omega$ is the matrix multiplication constant. 
To see this, observe that if $P$ is an adjacency matrix, $P^2$ counts paths of length 2 between each pair of vertices, so we may compute $|N_u^+ \cap N_v^+|$ for all $u,v$ by taking $P$ to be the adjacency matrix of positive edges. 
Phase (2) takes time $O(n^{2})$ (see Appendix \ref{sec: lp_rounding_alg}).
\end{proof}

Lastly, we will derive Corollary \ref{cor: apx-sparse}. 
For sparse graphs, instead of computing $d_{uv}$ for all pairs $u,v \in V$,
we need only compute $d_{uv}$ when $d_{uv} < 1$ (there are at most $n\Delta^2$ such pairs), and handle the other pairs implicitly.

\begin{proof}[Proof of Corollary \ref{cor: apx-sparse}]

For each $u \in V$, 
there are at most $\Delta^2$ vertices $v$ such that $|N_u^+ \cap N_v^+| > 0$, 
so there are at most $\Delta^2$ vertices $v$ with $d_{uv} < 1$. 
We only need to compute $d_{uv}$ for pairs $\{u,v\}$ such that $v$ in the 2-hop neighborhood $N^2(u)$ of $u$. For fixed $u$, computing $N^2(u)$ as well as $|N_u^+ \cap N_v^+|$ for each $v \in N^2(u)$ can be done in $O(\Delta^2)$ time, so for all vertices $u$ this takes a total of $O(n\Delta^2)$ time. 

Now we need to compute $d_{uv}$ for every $v \in N^2(u)$ 
(for $v \not \in N^2(u)$, $d_{uv} = 1$). 
This takes constant time for each pair $\{u,v\}$, 
since $$d_{uv} = 1 - |N_u^+ \cap N_v^+|/(n-|N_u^- \cap N_v^-|) = 1 - |N_u^+ \cap N_v^+|/\left(|N_u^+| + |N_v^+| - |N_u^+ \cap N_v^+|\right).$$

Now that we have $d$ computed for all relevant pairs of vertices,
it remains to argue the rounding algorithm (Algorithm \ref{KMZ-alg}) also can be run faster on sparse graphs. First, initialize $L_0(u)$ for each $u \in V$. Since $r = 1/5$, we only have to check whether $d_{uv} \leq 1/5$ for $v \in N^2(u)$ 
(otherwise, $d_{uv} = 1$). Since $|N^2(u)| \leq \Delta^2$, initializing $L_0(\cdot)$ takes $O(n\Delta^2)$ time total. We will store $L_t(u)$ for each $u \in V_t$ in a binary heap; inserting these takes $O(n\log n)$ time total. In each phase $t$, finding $u^*_t$ maximizing $L_t(u)$ will take $O(1)$ time. 
Then, we will have to remove all $v$ such that $d_{u_t^*v} \leq 2/5$. Since these $v$ must belong to $N^2(u^*_t)$, 
and each deletion is $O(\log n)$ time, this will take $O(\Delta^2 \log n)$ time. We will also have to decrease $L_{t}(u)$ to $L_{t+1}(u)$ for $u \in V_{t+1}$: 
For each vertex $v$ removed during phase $t$, 
there are at most $\Delta^2$ elements in $N^2(v)$, 
so $v$ induces at most $\Delta^2$ updates to $L_t(\cdot)$. Since each vertex is only removed once, 
the key update operations contribute $O(n\Delta^2 \log n)$ time total.
\end{proof}

\section{Faster Algorithm via Sampling}
\label{sec: sampling}
In this section, we show that instead of computing the \dn~exactly for each pair $u,v$, it suffices to estimate these via samples of the graph. In doing so, we can improve the run-time of our algorithm from $O(n^\omega)$ to $O(n^2 \log n)$. To do this, we show: (1) the fractional cost of the estimates is a constant factor away from \textsf{OPT} w.h.p. (Proposition \ref{prop: pseudo_vs_opt}), (2) the estimates satisfy an approximate triangle inequality w.h.p. (Proposition \ref{prop: pseudo_triangle}), and (3) inputting a function that approximately satisfies the triangle inequality to the rounding algorithm (Algorithm \ref{KMZ-alg}) is sufficient for obtaining a constant factor in (\midasterik) of line (\ref{eq: KMZ_chain}) (Lemma \ref{thm: constant_rounding_estimates}). These three steps imply Theorem \ref{thm: apx_sampling}. 

We will denote the initial estimates for $d_{uv}$ by $\bar{d}_{uv}$ 
and prove properties of the $\bar{d}_{uv}$ in Section \ref{sec: initial-est}. Then, we will have a post-processing phase, in which we round some $d_{uv}$ down to 0 and some up to 1; this step, while non-obvious, is needed to control the fractional cost. We examine the post-processing phase in Section \ref{sec: post_processed_estimates}. Our final inputs, which we use as input to the rounding algorithm, will be denoted $\widetilde{d}_{uv}$.

We will repeatedly use the following tail bounds for sums of random variables. 
\begin{theorem}[Chernoff-Hoeffding] \label{thm: chernoff}
Let $X = X_1 + \cdots + X_m$ where $\{X_1, \dots, X_m\}$ is a set of negatively correlated random variables. Define $\mu = \mathbb{E}[X]$. For $0 < \eps < 1$, the following tail bounds hold:
\[\mathbb{P}(X \geq (1+\eps)\mu) \leq e^{-\eps^2\mu/4} \]
\[\mathbb{P}(X \leq (1-\eps)\mu) \leq e^{-\eps^2\mu/4}. \]
\end{theorem}

\subsection{Initial estimates for the correlation metric}\label{sec: initial-est}
Fix $0 < \varepsilon < 1$. For every vertex $u$, randomly sample $ \lceil C(\varepsilon) \cdot \log n \rceil$ vertices from $|N_u^+|$, where $C(\varepsilon) = 32/ \eps^2$. Assume first that $|N_u^+| \geq \lceil C(\varepsilon) \cdot \log n \rceil $, 
as otherwise, the estimate is exact. Call this sample $S_u$. 

For two vertices $u,v$, define the random variables:
\begin{align*}
    X_{+}^{(u,v)} &= \sum_{w \in N_u^+ \cap N_v^+} \mathds{1}_{\{w \in S_u\}} \\
    X_{-}^{(u,v)} &= \sum_{w \in N_u^+ \cap N_v^-} \mathds{1}_{\{w \in S_u\}} \\
    W^{(u,v)} &= \frac{|N_u^+|}{\lceil C(\varepsilon)\log n \rceil} \cdot X_{+}^{(u,v)} \\
    Y^{(u,v)} &= \frac{|N_u^+|}{\lceil C(\varepsilon)\log n \rceil}\cdot X_{-}^{(u,v)}
\end{align*}
Note that the superscripts are \textit{ordered} pairs. If on the other hand $|N_u^+| < \lceil C(\eps) \cdot \log n \rceil$, set $W^{(u,v)} = |N_u^+ \cap N_v^+|$ and $Y^{(u,v)} = |N_u^+ \cap N_v^-|$.

Observe that $X_{+}^{(u,v)} + X_{-}^{(u,v)} = |S_u| = \lceil C(\varepsilon) \cdot \log n \rceil$. As a result,
\begin{equation} \label{eq: sum_estimates}
Y^{(u,v)} = |N_u^+| - W^{(u,v)}. 
\end{equation}
$Y^{(u,v)}$ will serve as the estimate for $|N_u^+ \cap N_v^-|$ and $W^{(u,v)}$ will serve as an estimate for $|N_u^+ \cap N_v^+|$. 

Flipping the order of $u$ and $v$ in the superscripts gives that $Y^{(v,u)}$ is an estimate of $|N_u^- \cap N_v^+|$, and $W^{(v,u)}$ is a second estimate for $|N_u^+ \cap N_v^+|$.
Also, observe that 
\begin{align}
    \mathbf{E}[X_{+}^{(u,v)}] &= |N_u^+ \cap N_v^+| \cdot \frac{\lceil C(\varepsilon)\log n \rceil}{|N_u^+|} \label{eq: pos_exp_val} \\
    \mathbf{E}[X_{-}^{(u,v)}] &= |N_u^+ \cap N_v^-| \cdot \frac{\lceil C(\varepsilon)\log n \rceil}{|N_u^+|} \label{eq: neg_exp_val} 
\end{align}
and similar statements for when the order of $u,v$ is flipped in the superscripts. \\
\textbf{Let $u,v$ be labelled so that $|N_v^+| \geq |N_u^+|$.} We define the \textit{initial} estimate for $d_{uv}$ (before post-processing) as 
\begin{equation} \label{eq: first_distance_estimate}
\widebar{d}_{uv} = \frac{Y^{(u,v)} + Y^{(v,u)}}{|N_u^+| + Y^{(v,u)}} = \frac{|N_u^+|- W^{(u,v)} + |N_v^+| - W^{(v,u)}}{|N_u^+| + |N_v^+| - W^{(v,u)}}
\end{equation}
where the second equality holds by equation (\ref{eq: sum_estimates}). Note this is well-defined since $|N_u^+| \geq 1$ due to $u$'s positive self-loop. Also note that in the denominator, we use the estimate $W^{(v,u)}$ of $|N_u^+ \cap N_v^+|$ from $v$'s sample, rather than $u$'s. This makes intuitive sense, since $|N_v^+| \geq |N_u^+|$, but the technical necessity for doing this will become clear later.

\begin{fact} \label{fact: init_estimate_bounds}
For every $u,v \in V$, $W^{(u,v)}, Y^{(u,v)}\geq 0$. Also, $0 \leq \widebar{d}_{uv} \leq 1$. 
\end{fact}

Since $|N_u^+| = |N_u^+ \cap N_v^+| + |N_u^+ \cap N_v^-|$, it is always the case that at least one of the two summands is at least $\frac{1}{2}|N_u^+|$. As a result, we can show in the next two propositions that at least one of the estimates of $|N_u^+ \cap N_v^+|$ and $|N_u^+ \cap N_v^-|$ is well-concentrated (Propositions \ref{prop: W_well_conc} and \ref{prop: Y_well_conc}). This concentration ensures that we are able to approximate $\bar{d}_{uv}$ with $d_{uv}$ 
from both above and below (Propositions \ref{prop: estimate_upper_bd} and \ref{prop: estimates_lower_bd}). 

\begin{proposition} \label{prop: W_well_conc}
Let $0 < \varepsilon < 1$. If $|N_u^+ \cap N_v^+| \geq \frac{1}{2}|N_u^+|$, then w.h.p. $W^{(u,v)} \geq (1-\varepsilon)|N_u^+ \cap N_v^+|$ and $W^{(u,v)} \leq (1+\varepsilon)|N_u^+ \cap N_v^+|$. Likewise, if $|N_u^+ \cap N_v^+| \geq \frac{1}{2}|N_v^+|$, then w.h.p. $W^{(v,u)} \geq (1-\varepsilon)|N_u^+ \cap N_v^+|$ and $W^{(v,u)} \leq (1+\varepsilon)|N_u^+ \cap N_v^+|$. In particular, each of the four events fails with probability at most $\frac{1}{n^{\eps^2 C(\eps)/8}}$.
\end{proposition}

\begin{proof}
If $|N_u^+| \leq \lceil C(\eps) \log n \rceil$, then the first two bounds hold automatically. So assume $|N_u^+| \geq \lceil C(\eps) \log n \rceil$. By Theorem \ref{thm: chernoff},
\begin{align*}
    \pr\left(W^{(u,v)} \leq (1-\varepsilon)|N_u^+ \cap N_v^+|\right) &= \pr\left(\frac{|N_u^+|}{\lceil C(\varepsilon)\log n \rceil}\cdot X_{+}^{(u,v)} \leq (1-\varepsilon)|N_u^+ \cap N_v^+| \right)\\
    &= \pr\left(X_{+}^{(u,v)} \leq (1-\varepsilon)|N_u^+ \cap N_v^+| \cdot \frac{\lceil C(\varepsilon) \log n \rceil}{|N_u^+|} \right) \\
    &= \pr\left(X_{+}^{(u,v)} \leq (1-\varepsilon)\mathbf{E}[X_u^+] \right) \\
    &\leq e^{-\frac{\varepsilon^2}{4}\mathbf{E}[X_{+}^{(u,v)}] } 
    \leq e^{-\frac{\varepsilon^2}{4} \cdot \lceil C(\varepsilon)\log n \rceil \cdot \frac{|N_u^+ \cap N_v^+|}{|N_u^+|}} 
    \leq e^{-\frac{\varepsilon^2}{8} \cdot C(\varepsilon) \log n},
\end{align*}
where we have used equation (\ref{eq: pos_exp_val}) and the assumption that $|N_u^+ \cap N_v^+| \geq \frac{1}{2}|N_u^+|$ in the last two inequalities. Similarly, 
\[\mathbf{P}\left(W^{(u,v)} \geq (1+\varepsilon)|N_u^+ \cap N_v^+| \right) \leq e^{-\frac{\varepsilon^2}{8} \cdot C(\varepsilon) \log n}.\]
\end{proof}

\begin{proposition} \label{prop: Y_well_conc}
Let $0 < \varepsilon < 1$. If $|N_u^+ \cap N_v^-| \geq \frac{1}{2}|N_u^+|$, then w.h.p. $Y^{(u,v)} \geq (1-\varepsilon)|N_u^+ \cap N_v^+|$ and $Y^{(u,v)} \leq (1+\varepsilon)|N_u^+ \cap N_v^+|$. Likewise, if $|N_u^- \cap N_v^+| \geq \frac{1}{2}|N_v^+|$, then w.h.p. $Y^{(v,u)} \geq (1-\varepsilon)|N_u^+ \cap N_v^+|$ and $Y^{(v,u)} \leq (1+\varepsilon)|N_u^+ \cap N_v^+|$. In particular, each of the four events fails with probability at most $\frac{1}{n^{\eps^2 C(\eps)/8}}$.
\end{proposition}

\begin{proof}
The proof is similar to the above, but instead using the assumption that $|N_u^+ \cap N_v^-| \geq (1/2)|N_u^+|$ and the fact that $\mathbf{E}[X_{-}^{(u,v)}] = \frac{|N_u^+ \cap N_v^-|}{|N_u^+|} \cdot \lceil C(\varepsilon)\log n \rceil$ by (\ref{eq: neg_exp_val}).  
\end{proof}

\bigskip

\begin{proposition} \label{prop: estimate_upper_bd}
For any $u,v \in V$, the following holds with probability at least $1-\frac{4}{n^{\eps^2C(\eps)/8}}$:
\[\widebar{d}_{uv} \leq c_1(\varepsilon) \cdot d_{uv} + h_1(\varepsilon) \]
and further, since $d_{uv}$ satisfy the triangle inequality, for any $u,v,w$,
\[\widebar{d}_{uv} \leq c_1(\varepsilon) \cdot (d_{uw} + d_{vw}) + h_1(\varepsilon) \]
where $c_1(\varepsilon) \rightarrow 1$ and $h_1(\varepsilon) \rightarrow 0$ as $\varepsilon \rightarrow 0$. In particular, take $c_1(\varepsilon) = (1+\varepsilon)/(1-\varepsilon)$ and $h_1(\varepsilon) = 2\varepsilon/(1-\varepsilon)$.
\end{proposition}

\begin{proof}
Assume WLOG that $|N_v^+| \geq |N_u^+|$. Let $A_{uv}$ denote the numerator of $d_{uv}$ and $B_{uv}$ denote the denominator, that is:
\begin{align*}
    A_{uv} &= |N_u^+ \cap N_v^-| + |N_u^- \cap N_v^+| \\
    B_{uv} &= |N_u^+ \cap N_v^-| + |N_u^- \cap N_v^+| + |N_u^+ \cap N_v^+|
\end{align*}
We case on whether or not $|N_u^+ \cap N_v^+| \geq (1/2)|N_u^+|$ and $|N_u^+ \cap N_v^+| \geq (1/2)|N_v^+|$. Note there are three possible cases instead of four, since $|N_u^+| \leq |N_v^+|$ makes it impossible that both $|N_u^+ \cap N_v^+| < (1/2)|N_u^+|$ and $|N_u^+ \cap N_v^+| \geq (1/2)|N_v^+|$ hold. 

\setcounter{case}{0}
\begin{case}
$|N_u^+ \cap N_v^+| \geq (1/2)|N_u^+|$ and $|N_u^+ \cap N_v^+| \geq (1/2)|N_v^+|$.
\end{case}
Then w.h.p. 
\begin{align*}
    \widebar{d}_{uv} &= \frac{|N_u^+|- W^{(u,v)} + |N_v^+| - W^{(v,u)}}{|N_u^+| + |N_v^+| - W^{(v,u)}} \\
    &\leq \frac{|N_u^+| - (1-\varepsilon)|N_u^+ \cap N_v^+| + |N_v^+| - (1-\varepsilon)|N_u^+ \cap N_v^+|}{|N_u^+| + |N_v^+| - (1+\varepsilon)|N_u^+ \cap N_v^+|} \\
    &= \frac{|N_u^+ \cap N_v^-| + |N_u^- \cap N_v^+| + 2\varepsilon |N_u^+ \cap N_v^+|}{|N_u^+ \cap N_v^+| + |N_u^+ \cap N_v^-| + |N_u^- \cap N_v^+| - \varepsilon |N_u^+ \cap N_v^+|} \\
    &= \frac{A_{uv}+2\varepsilon(B_{uv}-A_{uv})}{B_{uv} - \varepsilon(B_{uv}-A_{uv})}\\ 
    &\leq \frac{(1-2\varepsilon)A_{uv} + 2\varepsilon B_{uv}}{(1-\varepsilon)B_{uv}} 
    = \frac{1-2\varepsilon}{1-\varepsilon}\cdot d_{uv} + \frac{2\varepsilon}{1-\varepsilon},
\end{align*}
where we have used (\ref{eq: first_distance_estimate}) in the first line and Proposition \ref{prop: W_well_conc} in the second line. 

\begin{case}
$|N_u^+ \cap N_v^-| \geq \frac{1}{2}|N_u^+|$ and $|N_u^- \cap N_v^+| \geq \frac{1}{2}|N_v^+|$.
\end{case}
Then w.h.p.
\begin{align*}
    \widebar{d}_{uv} &= \frac{Y^{(u,v)} + Y^{(v,u)}}{|N_u^+| + Y^{(v,u)}} \\
    &\leq \frac{(1+\varepsilon)|N_u^+ \cap N_v^-| + (1+\varepsilon)|N_u^- \cap N_v^+|}{|N_u^+|+(1-\varepsilon)|N_u^- \cap N_v^+|} \\
    &\leq \frac{1+\varepsilon}{1-\varepsilon} \cdot d_{uv}
\end{align*}
where we have used (\ref{eq: first_distance_estimate}) in the first line and Proposition \ref{prop: Y_well_conc} in the second line.

\begin{case}
$|N_u^+ \cap N_v^+| \geq (1/2)|N_u^+|$ but $|N_u^+ \cap N_v^+| < (1/2)|N_v^+|$ (i.e., $|N_u^- \cap N_v^+| > (1/2)|N_v^+|$).
\end{case}
Then w.h.p. 
\begin{align*}
    \widebar{d}_{uv} &= \frac{|N_u^+|-W^{(u,v)} + Y^{(v,u)}}{|N_u^+|+Y^{(v,u)}} \\
    &\leq \frac{|N_u^+|-(1-\varepsilon)|N_u^+ \cap N_v^+| + (1+\varepsilon)|N_u^- \cap N_v^+|}{|N_u^+| + (1-\varepsilon)|N_u^- \cap N_v^+|} \\
    &\leq \frac{|N_u^+ \cap N_v^-| + \varepsilon |N_u^+ \cap N_v^+| + (1+\varepsilon)|N_u^- \cap N_v^+|}{(1-\varepsilon)B_{uv}} \\
    &\leq \frac{(1+\varepsilon)A_{uv}}{(1-\varepsilon)B_{uv}} + \frac{\varepsilon (B_{uv}-A_{uv})}{(1-\varepsilon)B_{uv}} \\
    &\leq \frac{1+\varepsilon}{1-\varepsilon}\cdot d_{uv} + \frac{\varepsilon}{1-\varepsilon}.
\end{align*}
where we have used (\ref{eq: first_distance_estimate}) in the first line and Propositions \ref{prop: W_well_conc} and \ref{prop: Y_well_conc} in the second line. 
\end{proof}

\bigskip

\begin{proposition} \label{prop: estimates_lower_bd}
For any $u,v \in V$, the following holds with probability at least $1-\frac{4}{n^{\eps^2C(\eps)/8}}$:
\[d_{uv} \leq c_2(\varepsilon)\cdot\widebar{d}_{uv} + h_2(\varepsilon) \]
where $c_2(\varepsilon) \rightarrow 1$ and $h_2(\varepsilon) \rightarrow 0$ as $\varepsilon \rightarrow 0$. In particular, take $c_2(\varepsilon) = (1+\varepsilon)/(1-\varepsilon)$ and $h_2(\varepsilon) = 2\varepsilon/(1-\varepsilon)$.
\end{proposition}

\begin{proof}
The proof follows as that of Proposition \ref{prop: estimate_upper_bd}, 
but using the other side of the Chernoff bound.
\end{proof}

We combine the last two propositions to show our sampling leads
to an approximate triangle inequality. 

\begin{proposition} \label{prop: tri_single_triplet}
For any $u,v,w \in V$, the following holds with probability at least $1-\frac{12}{n^{\eps^2C(\eps)/8}}$:
\[\widebar{d}_{uv} \leq c_3(\varepsilon)\cdot(\widebar{d}_{uw} + \widebar{d}_{vw}) + h_3(\varepsilon) \]
where $c_3(\varepsilon) \rightarrow 1$ and $h_3(\varepsilon) \rightarrow 0$ as $\varepsilon \rightarrow 0$. In particular, take 
\[c_3(\varepsilon) = \left(\frac{1+\varepsilon}{1-\varepsilon}\right)^2 \text{ and } \quad
h_3(\varepsilon) = \left(\frac{2\varepsilon}{1-\varepsilon}\right)\left(1+\frac{2(1+\varepsilon)}{1-\varepsilon}\right).\]
\end{proposition}
\begin{proof}
This follows directly from Propositions \ref{prop: estimate_upper_bd} and \ref{prop: estimates_lower_bd}.
\end{proof}

Then, we take a union bound to show the approximate triangle inequality holds globally for all pairs with high probability.

\begin{proposition} \label{prop: intermediate_tri}
Fix $\varepsilon > 0$. An approximate triangle inequality $\widebar{d}_{uv} \leq c_3(\varepsilon)\cdot(\widebar{d}_{uw} + \widebar{d}_{vw}) + h_3(\varepsilon)$ holds for all triplets $u,v,w$ \textbf{simultaneously} with probability at least $1-O\left(\frac{1}{n}\right)$.
\end{proposition}

\begin{proof}
The triangle inequality fails for an arbitrary triplet $u,v,w$ with probability at most $\frac{12}{n^{\eps^2C(\eps)/8}}$ by Proposition \ref{prop: tri_single_triplet}. As there are at most $n^3$ triplets, the triangle inequality  fails on at least one triplet with probability at most $\frac{12}{n^{\eps^2\frac{C(\eps)}{8} - 3}} = \frac{12}{n}$ since $C(\eps) = 32/\eps^2$. 
\end{proof}

\subsection{Post-processed estimates for the correlation metric} \label{sec: post_processed_estimates}

Now we define our final estimates $\widetilde{d}_{uv}$:
\[\widetilde{d}_{uv} =
  \begin{cases} 
      0 & \widebar{d}_{uv} \leq 2h_3(\varepsilon) \text{ and }(u,v) \in E^+  \\
      1 & \widebar{d}_{uv} \geq 1/(2c_3(\varepsilon)+1) \text{ and }(u,v) \in E^- \\
      \widebar{d}_{uv} & \text{otherwise}
   \end{cases}.
\]

In order to show that our algorithm is successful using $\widetilde{d}_{uv}$ in place of the exact distances $d_{uv}$, we need to demonstrate that w.h.p. (1) $\widetilde{d}_{uv}$ satisfy an approximate triangle inequality (Proposition \ref{prop: pseudo_triangle}), and (2) the fractional cost of $\widetilde{d}_{uv}$ can be compared to \textsf{OPT} (Proposition \ref{prop: pseudo_vs_opt}). In particular, rounding the estimates to $\widetilde{d}$ allows us to trade a small loss in the approximate triangle inequality that $\widebar{d}$ satisfied for the ability to bound the fractional cost.  

\begin{proposition} \label{prop: pseudo_triangle}
Fix $\varepsilon > 0$. An approximate triangle inequality holds for all triplets $u,v,w$ simultaneously with probability at least $1-O(1/n)$:
\[\widetilde{d}_{uv} \leq c_4(\varepsilon) \cdot \left(\widetilde{d}_{uw} + \widetilde{d}_{vw} \right) + h_4(\varepsilon),\]
where $h_4(\eps) \rightarrow 0$ and $c_4(\eps) \rightarrow 3$ as $\eps \rightarrow 0$. In particular, take $h_4(\eps) = (4c_3(\eps) + 1)(2c_3(\eps)+1)h_3(\eps)$ and $c_4(\eps) = (2c_3(\eps)+1)c_3(\eps)$. 
\end{proposition}

\begin{proposition} \label{prop: pseudo_vs_opt}
Fix $0 < \eps < 0.03$. For $D(\varepsilon) = 2c_3(\eps)$, the following inequality holds with probability at least $1-O(1/n)$:
\[\sum_{v \in N_u^+} \widetilde{d}_{uv} + \sum_{v \in N_u^-} (1- \widetilde{d}_{uv}) \leq D(\varepsilon) \cdot \left(\sum_{v \in N_u^+} d_{uv} + \sum_{v \in N_u^-} (1- d_{uv})\right) \leq 8 \cdot D(\varepsilon) \textsf{OPT}.\]
\end{proposition}

\bigskip

\begin{proof}[Proof of Proposition \ref{prop: pseudo_triangle}]
By Proposition \ref{prop: intermediate_tri}, 
\begin{equation} \label{eq: intermediate_tri}
\widebar{d}_{uv} \leq c_3(\varepsilon)\cdot(\widebar{d}_{uw} + \widebar{d}_{vw}) + h_3(\varepsilon) 
\end{equation}
holds simultaneously for all triplets $u,v,w$ w.h.p. We need to show that a similar inequality holds when we replace the intermediate estimates $\widebar{d}(\cdot)$ with the post-processed estimates $\widetilde{d}(\cdot)$. First observe that we round $\widebar{d}_{uv}$ up to 1 only when $\widebar{d}_{uv} \geq 1/(2c_3(\varepsilon)+1)$. So in this case
\[\widetilde{d}_{uv} = 1 = \left(1/(2c_3(\varepsilon)+1)\right)(2c_3(\varepsilon)+1) \leq \widebar{d}_{uv} \cdot (2c_3(\varepsilon)+1).\]
So it always true that
\begin{equation} \label{eq: pseudo_tri_step1}
\widetilde{d}_{uv} \leq \widebar{d}_{uv} \cdot (2c_3(\varepsilon)+1)
\end{equation}
since this inequality also holds for $\widetilde{d}_{uv} = 0$ and for $\widetilde{d}_{uv} =\widebar{d}_{uv}$ (using $\widebar{d}_{uv} \geq 0$ by Fact \ref{fact: init_estimate_bounds}). 
Next observe that we round $\widebar{d}_{uw}$ down to 0 only when $\widebar{d}_{uw} \leq 2h_3(\eps)$. So it is always true that 
\begin{equation} \label{eq: pseudo_tri_step2}
\widebar{d}_{uw} - 2h_3(\eps) \leq \widetilde{d}_{uw}.
\end{equation}
This is because (1) if $\widebar{d}_{uw}$ is rounded up to $\widetilde{d}_{uw} = 1$, then the inequality holds since $\widebar{d}_{uw} \leq 1$ (Fact \ref{fact: init_estimate_bounds}), (2) if $\widetilde{d}_{uw} =\widebar{d}_{uw} $, then the inequality holds automatically, and (3) if $\widebar{d}_{uw}$ is rounded down to $\widetilde{d}_{uw} = 0$ then $\widebar{d}_{uw} \leq 2h_3(\eps)$, so again the inequality holds. \\

Putting inequalities (\ref{eq: intermediate_tri}), (\ref{eq: pseudo_tri_step1}), and (\ref{eq: pseudo_tri_step2}) together gives:
\begin{align*}
\widetilde{d}_{uv} &\leq (2c_3(\eps)+1) \widebar{d}_{uv}\\
&\leq(2c_3(\eps)+1)c_3(\varepsilon)(\widebar{d}_{uw} + \widebar{d}_{vw}) + (2c_3(\eps)+1)h_3(\varepsilon) \\
&\leq (2c_3(\eps)+1)c_3(\varepsilon)\left(\widetilde{d}_{uw} + \widetilde{d}_{vw} + 4h_3(\eps)\right) + (2c_3(\eps)+1)h_3(\varepsilon)\\
&\leq c_4(\eps)\left(\widetilde{d}_{uw} + \widetilde{d}_{vw} \right) + h_4(\eps)
\end{align*}
where $c_4(\eps) = (2c_3(\eps)+1)c_3(\eps)$ and $h_4(\eps) = (4c_3(\eps) + 1)(2c_3(\eps) + 1)h_3(\eps)$. In particular, $c_4(\eps) \rightarrow 3$ and $h_4(\eps) \rightarrow 0$ as $\eps \rightarrow 0$. 

\end{proof}

\bigskip

\begin{proof}[Proof of Proposition \ref{prop: pseudo_vs_opt}]
It suffices to show the bound holds pointwise, e.g., for $v \in N_u^+$, we show that $\widetilde{d}_{uv} \leq D(\varepsilon) \cdot d_{uv}$ and for $v \in N_u^-$, we show that $(1-\widetilde{d}_{uv}) \leq D(\varepsilon)(1-d_{uv})$.\\ 

First consider $v \in N_u^+$. If $\widebar{d}_{uv}$ was rounded down to $\widetilde{d}_{uv} = 0$, then $\widetilde{d}_{uv} \leq D(\varepsilon) d_{uv}$ holds automatically. If $\widebar{d}_{uv}$ was not rounded down, then $\widetilde{d}_{uv} = \widebar{d}_{uv} \geq 2h_3(\varepsilon)$. By Proposition \ref{prop: estimate_upper_bd},  w.h.p.
\[ 2h_3(\varepsilon) \leq \widetilde{d}_{uv} = \widebar{d}_{uv} \leq c_3(\varepsilon)d_{uv} + h_3(\varepsilon)\]
where we have used that $c_1(\varepsilon) \leq c_3(\varepsilon)$ and $h_1(\varepsilon) \leq h_3(\varepsilon)$. So $d_{uv} \geq h_3(\varepsilon)/c_3(\varepsilon)$. In turn, 
\[\widetilde{d}_{uv} = \widebar{d}_{uv} \leq c_3(\varepsilon)d_{uv} + h_3(\varepsilon) \leq c_3(\varepsilon)d_{uv} + c_3(\varepsilon)d_{uv} \leq D(\varepsilon) \cdot d_{uv}   \]
as desired.

 Next consider when $v \in N_u^-$. This case will be more involved. If $\widebar{d}_{uv}$ was rounded up to $\widetilde{d}_{uv} = 1$, then $ 0 = 1-\widetilde{d}_{uv} \leq D(\varepsilon) \cdot (1-d_{uv})$ holds automatically. Otherwise, $\widetilde{d}_{uv} = \widebar{d}_{uv} \leq 1/(2c_3(\varepsilon)+1)$. We consider two cases:

\setcounter{case}{0}
\begin{case}
$h_3(\varepsilon) \leq \widebar{d}_{uv}$.
\end{case}
By Proposition \ref{prop: estimates_lower_bd}, w.h.p.
\[c_3(\varepsilon) \cdot \widebar{d}_{uv} + h_3(\varepsilon) \geq d_{uv} \]
where we have used that $c_2(\varepsilon) \leq c_3(\varepsilon)$ and $h_2(\varepsilon) \leq h_3(\varepsilon)$. 
Now using the assumptions of this case, 
\begin{align}
\widebar{d}_{uv} &\geq \frac{1}{c_3(\eps)} \cdot d_{uv} - \frac{h_3(\eps)}{c_3(\eps)} \notag \\
\widebar{d}_{uv} &\geq \frac{1}{c_3(\eps)} \cdot d_{uv} - \frac{\widebar{d}_{uv}}{c_3(\eps)} \notag \\
\left(1 + \frac{1}{c_3(\eps)} \right) \cdot \widebar{d}_{uv} &\geq \frac{1}{c_3(\eps)} \cdot d_{uv} \notag \\
\widebar{d}_{uv} &\geq \frac{1}{c_3(\eps) + 1} \cdot d_{uv} \label{mult_error}
\end{align}
Since $\widebar{d}_{uv} \leq 1/(2c_3(\eps)+1)$, it holds that 
\[1-\widebar{d}_{uv} \leq 2 \left(1- (c_3(\eps)+1)\widebar{d}_{uv} \right).\]
Now, using (\ref{mult_error}), 
\[1-\widetilde{d}_{uv} = 1-\widebar{d}_{uv} \leq 2(1-d_{uv}) \leq D(\varepsilon) \cdot (1-d_{uv}). \]

\begin{case}
$h_3(\varepsilon) \geq \widebar{d}_{uv}$.
\end{case}
As in the previous case we use Proposition \ref{prop: estimates_lower_bd}. This, along with the assumptions of this case gives, w.h.p.
\begin{align*}
d_{uv} &\leq c_3(\varepsilon) \cdot \widebar{d}_{uv} + h_3(\varepsilon) \\
d_{uv} &\leq (c_3(\eps) +1)h_3(\eps) \\
1 - d_{uv} &\geq 1 - (c_3(\eps) +1)h_3(\eps).
\end{align*}
Now since $c_3(\varepsilon) \rightarrow 1$ and $h_3(\varepsilon) \rightarrow 0$, we have that for small enough $\eps$ ($\eps < 0.03$ suffices), $(c_3(\eps) +1)h_3(\eps) \leq 1/2$. In turn, 
\begin{align*}
1-d_{uv} &\geq 1/2 \\
1-d_{uv} &\geq 1/2(1-\widebar{d}_{uv})
\end{align*}
where we have used that $\widebar{d}_{uv} \geq 0$ (Fact \ref{fact: init_estimate_bounds}). So 
\[1-\widetilde{d}_{uv} = 1-\widebar{d}_{uv} \leq 2(1-d_{uv}) \leq D(\varepsilon) \cdot (1-d_{uv})\]
which concludes the case and the proof.
\end{proof}

\subsection{Approximate triangle inequality suffices}
\label{sec: approx_tri_suff}

In this section, 
we show that instead of inputting a semi-metric to the rounding algorithm 
(Algorithm \ref{KMZ-alg}), 
one can use as input a function that is almost a semi-metric. 
We will call such a function $d$ a \textit{$(\delta_1,\delta_2)$-semi-metric}
if it is a semi-metric except instead of satisfying the 
triangle inequality, it satisfies a $(\delta_1, \delta_2)$-approximate triangle inequality:

\begin{definition}
The function $d: V^2 \rightarrow \mathbb{R}_{\geq 0}$ satisfies a
\textit{$(\delta_1, \delta_2)$-approximate triangle inequality}  when 
$$d_{uv} \leq \delta_1(d_{uw} + d_{wv}) + \delta_2 \hspace{0.3cm} \forall u,v,w \in V.$$
\end{definition}
\begin{algorithm} \label{algo}
\end{algorithm}

\begin{minipage}{14cm}
\noindent \textbf{Input:} $d$ a $(\delta_1, \delta_2)$-semi-metric on $V$

\noindent \textbf{Output: }Clustering $\mathcal{C}$.
\begin{enumerate}
    \item Let $V_0 = V$, $r=r(\delta_1,\delta_2)$, $t=0$.
    \item \textbf{while} ($V_t \neq \emptyset$)
    \begin{itemize}
        \item Find $w_t = \arg \max_{w \in V_t} L_t(w) =  \arg \max_{w \in V_t} \sum_{u \in \text{Ball}(w,r) \cap V_t} {r-d_{uw}}$. 
        \item Create ${C_t = \text{Ball}(w_t,b \cdot r) \cap V_t}$, for $b=b(\delta_1, \delta_2)$.
        \item Set $V_{t+1} = V_t \setminus C_t$ and $t=t+1$.
    \end{itemize}
    \item Return ${\mathcal{C} = (C_0, \dots, C_{t-1})}$.
\end{enumerate}
\end{minipage}

\bigskip 

Recall as in Section \ref{sec: overviews}
that $\widehat{d}_{uv} = d_{uv}$ if $(u,v) \in E^+$,
$\widehat{d}_{uv} = 1 - d_{uv}$ if $(u,v) \in E^-$, 
and $\textsf{ALG}(u,v) = \mathds{1}((u,v) \text{ is in disagreement in }\mathcal{C})$, for $\mathcal{C}$ the clustering returned by Algorithm \ref{algo}.

Let $\eps > 0$ be sufficiently small. By Proposition \ref{prop: pseudo_triangle}, $\widetilde{d}$ satisfies a $(\delta_1, \delta_2)$-approximate triangle inequality, where $\delta_1 = \delta_1(\eps) =  3 + h_4(\eps)$ and $\delta_2 = \delta_2(\eps) = h_4(\eps)$. (This follows by noting that $c_4(\eps) \leq 3 + h_4(\eps)$. Thus the $(\delta_1, \delta_2)$-approximate triangle inequality is weaker than that in Proposition \ref{prop: pseudo_triangle}, but we use the former for ease of computation.)
\begin{lemma}[Analogue of Theorem B.1 in \cite{KMZ19}] \label{thm: constant_rounding_estimates}
 Let $r = r(\delta_1, \delta_2)$ and $b = b(\delta_1, \delta_2)$ be as defined in Appendix \ref{appendix: approx_tri_requirements}.\footnote{For us, $\delta_1(\eps) \rightarrow 3$ and $\delta_2(\eps) \rightarrow 0$ as $\eps \rightarrow 0$. A similar result can be obtained when $\delta_1$ approaches an arbitrary constant, as long as $\delta_2(\eps) \rightarrow 0$ as $\eps \rightarrow 0$.} Let $\mathcal{C}$ be a clustering returned by Algorithm \ref{algo}. 
Then,
$$\textsf{ALG}(u) = \sum_{v} \textsf{ALG}(u,v) \leq \frac{1}{r(\delta_1, \delta_2)} \sum_v \widehat{d}_{uv}$$
where $r(\delta_1, \delta_2) \rightarrow 1/121$.\footnote{If $\delta_1(\eps) \rightarrow 1$, then $r(\delta_1, \delta_2) \rightarrow 1/5$ as $\eps \rightarrow 0$, recovering the rounding result that holds when the exact triangle inequality is satisfied.} 
\end{lemma}

\begin{proof}
\noindent We will follow the proof of Theorem B.1. However, the change in the parameter $r$ and in how $C_t$ is created (i.e., the radius to $b\cdot r$ instead of $2r$) will cause the cases to split at different points. Note that $b \cdot r < 1$. \\

\noindent Define 
$\text{profit}(u) = \sum_v \widehat{d}_{uv} - r \sum_v \textsf{ALG}(u,v)$ 
    and
$\Delta E_t = \{(u,v) \mid u \in C_t \text{ or }v \in C_t \}.$
Then, 
\[\text{profit}_t(u,v) =  \begin{cases} 
      \widehat{d}_{uv} - r \textsf{ALG}(u,v) & (u,v) \in \Delta E_t \\
      0 & o.w.
   \end{cases}
\]
and
\[\text{profit}_t(u) = \sum_{v \in V_t} \profit_t(u,v) = \sum_{(u,v) \in \Delta E_t, v \in V_t} \widehat{d}_{uv} - r \sum_{(u,v) \in \Delta E_t, v \in V_t} \textsf{ALG}(u,v). \]
Observe that since $\Delta E_t$ are disjoint, 
$\profit(u) = \sum_t \profit_t(u).$ 
Note that if $u \not \in V_t$, $\profit_t(u) = 0$. 

\bigskip

\begin{lemma}[Analogue of Lemma B.2 in \cite{KMZ19}] \label{lem: analogue_b2}
For every $u \in V_t$, $\profit_t(u) \geq 0$. Consequently, $\profit(u) \geq 0$, giving Lemma \ref{thm: constant_rounding_estimates}. 
\end{lemma}
Take $c_1 = c_1(\delta_1,\delta_2)$ and $c_2 = c_2(\delta_1,\delta_2)$ as defined in Appendix \ref{appendix: approx_tri_requirements}. Note $c_2 \cdot r<1$ and $c_1 < c_2$, so the cases in the proof of Lemma \ref{lem: analogue_b2} make sense. 
\begin{proof}
Let $w = w_t$ be the maximizer of $L_t$. 

\begin{case} 
($d_{uw} \in [0,c_1 \cdot r] \cup [c_2 \cdot r,1]$)
\end{case}

\noindent In this case, it suffice to show that $\profit_t(u,v) \geq 0$ for $v \in V_t$ and $(u,v) \in \Delta E_t$. This follows by Claim \ref{B3} and Lemmas \ref{B4} and \ref{B5}.

\begin{case} 
($d_{uw} \in (c_1 \cdot r, c_2 \cdot r)$)
\end{case}

\noindent By Lemmas \ref{B7} and \ref{B9}, $\profit_t(u) \geq L_t(w) - L_t(u) \geq 0$, since $w$ is the maximizer of $L_t$.

\end{proof}

\begin{claim}[Analogue of Claim B.3 in \cite{KMZ19}] \label{B3}
Let $u,v \in V_t$ and $(u,v) \in E^-$. Then $\profit_t(u,v) \geq 0$.
\end{claim}

\begin{proof}
Note the claim holds automatically if 
$\textsf{ALG}(u,v) = 0$ or $(u,v) \not \in \Delta E_t$. 
So we consider when $\textsf{ALG}(u,v) = 1$ and $(u,v) \in \Delta E_t$. 
In this case, it must be that both $u$ and $v$ are in $C_t$, 
since $(u,v) \in E^- \cap \Delta E_t$. 
By definition of $C_t$, $d_{uw} \leq b \cdot r$ and 
$d_{vw} \leq b \cdot r$. 
By the approximate triangle inequality, $d_{uv} \leq 2 \delta_1 b \cdot r+\delta_2$. 
This gives 
\begin{align*}
    \profit_t(u,v) &= \widehat{d}_{uv} - r \cdot \textsf{ALG}(u,v) = 1-d_{uv} - r   \geq 1-2 \delta_1 b\cdot r - \delta_2-r \geq 0,
\end{align*}
where the last inequality follows by the choice of constants.
\end{proof}

\begin{lemma}[Analogue of Lemma B.4 in \cite{KMZ19}] \label{B4}
Let $u \in V_t$ and $d_{uw} \in  [0,c_1 \cdot r]$.
Then $\profit_t(u,v) \geq 0$ for all $v \in V_t$ and $(u,v) \in \Delta E_t$. 
\end{lemma}

\begin{proof}
Due to Claim \ref{B3}, we may assume $(u,v) \in E^+$. 
Since $d_{uw} \leq c_1 \cdot r$ and $c_1 \leq b$, we have $u \in \text{Ball}(w,b \cdot r)$,
thus $u \in C_t$. 
So $(u,v) \in E^+$ is a disagreement if and only if $v \not \in C_t$, 
i.e., if and only if $d_{vw} > b \cdot r$.
By our choice of constants
\begin{align*}
    \profit_t(u,v) &\geq \widehat{d}_{uv} - r \textsf{ALG}(u,v) \geq d_{uv} - r 
    \geq \frac{1}{\delta_1}d_{vw} - d_{uw} - r - \delta_2/\delta_1\\
    & \geq \frac{1}{\delta_1} b \cdot r - c_1 \cdot r -r - \delta_2/\delta_1 \geq 0.
\end{align*}
\end{proof}

\begin{lemma}[Analogue of Lemma B.5 in \cite{KMZ19}] \label{B5}
Let $u \in V_t$ and $d_{uw} \in [c_2 \cdot r,1]$. 
Then $\profit_t(u,v) \geq 0$ for all $v \in V_t$ and $(u,v) \in \Delta E_t$. 
\end{lemma}

\begin{proof}
Due to Claim \ref{B3}, we may assume $(u,v) \in E^+$. 
Since $d_{uw} \geq c_2 \cdot r$ and $c_2 > b$, we have  
$u \not \in \text{Ball}(w, b \cdot r) \cap V_t = C_t$.
In turn $(u,v) \in E^+$ is a disagreement if and only if $v \in C_t$ 
(note $(u,v) \in \Delta E_t$ by assumption), 
i.e., $d_{vw} \leq b \cdot r$. 
By the choice of constants, this gives
\begin{align*}
    \profit_t(u,v) &= d_{uv} - r \geq \frac{1}{\delta_1} d_{uw} - d_{vw} - r - \delta_2/\delta_1
       \geq \frac{1}{\delta_1 }c_2 \cdot r - b \cdot r - r - \delta_2/\delta_1\geq 0.
\end{align*}
\end{proof}

\begin{claim}[Analogue of Claim B.6 in \cite{KMZ19}] \label{B6}
Let $u \in V_t$, 
$d_{uw} \in (c_1 \cdot r,c_2 \cdot r]$,
and $v \in \text{Ball}(w,r) \cap V_t$. 
Then $\profit_t(u,v) \geq r - d_{vw}$. 
\end{claim}

\begin{proof}
\begin{case}
$d_{uw} \in (c_1 \cdot r, b \cdot r]$.
\end{case}
\noindent In this case, $u \in C_t$, since $d_{uw} \leq b \cdot r$. 
Also, since $v \in \text{Ball}(w,r) \cap V_t$ and $b\geq 1$, we have 
$v \in \text{Ball}(w,b \cdot r) \cap V_t$, so $v \in C_t$ as well. 
Thus, if $(u,v) \in E^+$, $ALG(u,v) = 0$, and by the choice of constants
\begin{align*}
    \profit_t(u,v) &= \widehat{d}_{uv} - r \textsf{ALG}(u,v) = d_{uv} - 0 
    \geq \frac{1}{\delta_1}d_{uw} - d_{vw} - \delta_2/\delta_1 \\
    & \geq \frac{1}{\delta_1}c_1 \cdot r - d_{vw} - \delta_2/\delta_1
    = r - d_{vw}.
\end{align*}

On the other hand, if $(u,v) \in E^-$, then 
\begin{align*}
    \profit_t(u,v) &= \widehat{d}(u,v) - r \cdot \textsf{ALG}(u,v) 
     \geq 1-d_{uv} - r 
    \geq 1 - \delta_1d_{uw} - \delta_1d_{vw} - r - \delta_2\\
    &= 1 - \delta_1d_{uw} - r - (\delta_1-1)d_{vw} - d_{vw} - \delta_2 \\
    &= 1 - (\delta_1 b +\delta_1) r - d_{vw} - \delta_2
    \geq r - d_{vw}
\end{align*}

\begin{case}
$d_{uw} \in (b \cdot r,c_2 \cdot r]$.
\end{case}
\noindent If $(u,v) \in E^+$, then 
\begin{align*}
    \profit_t(u,v) &= \widehat{d}_{uv} - r\textsf{ALG}(u,v) 
    \geq d_{uv} - r 
    \geq \frac{1}{\delta_1}d_{uw} - d_{vw} - r - \delta_2/\delta_1\\
    & =  (\frac{1}{\delta_1} b-1)r - d_{vw} - \delta_2/\delta_1
    \geq r - d_{vw}.
\end{align*}

\noindent If $(u,v) \in E^-$, then 
\begin{align*}
    \profit_t(u,v) &= \widehat{d}_{uv} - r\textsf{ALG}(u,v) 
    \geq 1 - d_{uv} - r 
    \geq 1 - \delta_1d_{uw} - \delta_1d_{vw} - r - \delta_2\\
    &\geq 1 - (\delta_1 c_2 + \delta_1)r - d_{vw} - \delta_2
    \geq r - d_{vw},
\end{align*}
where we have used that $\delta_1 \geq 3$ implies $1-\delta_1 < 0$. 
\end{proof}

\begin{lemma}[Analogue of Lemma B.7 in \cite{KMZ19}] \label{B7}
If $d_{uw} \in (c_1 \cdot r,c_2 \cdot r]$, 
then $P_{high}(u) \geq L_t(w)$. 
\end{lemma}

\begin{proof}
By Claim \ref{B6}, $\profit_t(u,v) \geq r - d_{vw}$ for all $v \in \text{Ball}(w,r) \cap V_t$. So
\[P_{high}(u) = \sum_{v \in \text{Ball}(w,r) \cap V_t} \profit_t(u,v) \geq \sum_{v \in \text{Ball}(w,r) \cap V_t} r - d_{vw} = L_t(w). \]
\end{proof}

\begin{claim}[Analogue of Claim B.8 in \cite{KMZ19}] \label{B8}
Let $u,v \in V_t$. Then $\profit_t(u,v) \geq \min(d_{uv}-r, 0).$
\end{claim}

\begin{proof}
If $(u,v) \in E^-$, then $\profit_t(u,v) \geq 0$ by Claim \ref{B3}. If $(u,v) \not \in \Delta E_t$, then $\profit_t(u,v) = 0$ by definition. So we may assume $(u,v) \in E^+ \cap \Delta E_t$. In this case, 
\begin{align*}
    \profit_t(u,v) &= \widehat{d}_{uv} - r \textsf{ALG}(u,v) 
    \geq d_{uv} - r 
    \geq \min(d_{uv}-r,0).
\end{align*}
\end{proof}

\begin{lemma}[Analogue of Lemma B.9 in \cite{KMZ19}] \label{B9}
Let $u \in V_t$. Then $P_{low}(u) \geq -L_t(u)$, where $P_{low}(u) = \sum_{v \in V_t \setminus \text{Ball}(w,r)} \profit_t(u,v)$.
\end{lemma}

\begin{proof}
\begin{align*}
    P_{low}(u) &= \sum_{v \in V_t \setminus \text{Ball}(w,r)} \profit_t(u,v) 
    \geq \sum_{v \in V_t \setminus \text{Ball}(w,r)} \min(d_{uv} - r, 0) \\
    &\geq \sum_{v \in V_t} \min(d_{uv} - r, 0) 
    = \sum_{v \in \text{Ball}(u,r) \cap V_t} d_{uv}-r  
    = -L_t(u)
\end{align*}
where in the second line we have used Claim \ref{B8}, in the third line we have used that all terms are non-positive, and in the fourth line we have used that $\min(d_{uv} - r,0) = 0$ if $v \not \in \text{Ball}(u,r)$. 
\end{proof}
\end{proof}

\subsection{Completing the proof of Theorem \ref{thm: apx_sampling}} \label{sec: sampling_thm_pf}
This section ties together the analysis of the previous sections to prove that sampling can be used to improve the run-time of our combinatorial algorithm for Min Max correlation clustering while still obtaining a constant factor approximation.

\begin{proof}[Proof of Theorem \ref{thm: apx_sampling}]
Our sampling algorithm for correlation clustering with respect to the $\ell_{\infty}$ norm is the following. Compute $\widetilde{d}$ via sampling as described in Section \ref{sec: post_processed_estimates}. Feed $\widetilde{d}$ as input to Algorithm \ref{KMZ-alg}. Let $\textsf{ALG}(u,v) = \mathds{1}((u,v) \text { is in disagreement in }\mathcal{C})$ 
and $\textsf{ALG}(u) = \sum_{v \in V} \textsf{ALG}(u,v)$. Define $\widehat{\widetilde{d}}_{uv} = \widetilde{d}_{uv}$ if $v \in N_u^+$ and $\widehat{\widetilde{d}}_{uv} = 1 - \widetilde{d}_{uv}$ if $u \in N_v^-$. Following line (\ref{eq: KMZ_chain}), we see that, with probability $1-O(1/n)$, 
\begin{equation} 
    \textsf{ALG}(u) = \sum_{v \in V}\textsf{ALG}(u,v)  \underset{*}{\leq} 
\frac{1}{r(\eps)} \cdot \sum_{v \in V}\widehat{\widetilde{d}}_{uv} \underset{**} \leq \frac{1}{r(\eps)} \cdot D(\eps) \cdot \textsf{OPT}.
\end{equation} 
where $D(\eps) = 2 \cdot \left(\frac{1+\eps}{1-\eps} \right)^2$ and $r(\eps) = r(\delta_1, \delta_2)$ (recall that $\delta_1=\delta_1(\epsilon)$ and $\delta_2=\delta_2(\epsilon)$), as defined in Appendix \ref{appendix: approx_tri_requirements}. The inequality (\midasterik) follows from Proposition \ref{prop: pseudo_triangle} and Lemma \ref{thm: constant_rounding_estimates},  and the inequality (\midasterik \midasterik) follows from Proposition \ref{prop: pseudo_vs_opt}. 

Next, we analyze the run-time. As before, there are two phases: (1) computing the estimates $\widetilde{d}_{uv}$, and (2) using the rounding algorithm (Algorithm \ref{KMZ-alg}) with input $\widetilde{d}$. Phase (1) takes $O(n^2\log n/\eps^2)$ time. First we compute the sample $S_u$ for each vertex $u$, which takes $O(n \log n/\eps^2)$ time, since we sample $O(\log n/\eps^2)$ vertices from $N_u^+$. Then, for each of the $O(n^2)$ pairs $u,v$, we compute $W^{(u,v)}$ and $Y^{(u,v)}$, which takes $O(\log n/\eps^2)$ time. Thus to compute $\bar{d}_{uv}$ as in (\ref{eq: first_distance_estimate}) for all pairs takes $O(n^2 \log n / \eps^2)$ time. Finally, obtaining $\widetilde{d}$ from $\widebar{d}$ via rounding takes $O(n^2)$ time. Phase (2) takes $O(n^2)$ time, by the discussion in Appendix \ref{sec: lp_rounding_alg}. Together the two phases contribute a total run-time of $O(n^2\log n/\eps^2)$, completing the proof of Theorem \ref{thm: apx_sampling}.
\end{proof}

\section{Experiments} \label{sec: experiments}
In this section, we describe experiments supporting our theoretical results.\footnote{The code for our experiments can be found at https://github.com/hanewman/MinMax-Correlation-Clustering-} We demonstrate:
\begin{itemize}
\item The guarantees of Theorem \ref{thm: apx_thm} are predictive of 
our algorithm's performance on real-world and synthetic datasets. Our algorithm's solution quality is similar to the KMZ algorithm. 
\item The fractional cost of the correlation metric is similar to the objective value of LP \ref{KMZ_LP}. 
\item Our algorithm is \textit{scalable}: we can handle large graphs (up to $\approx 10,000$ vertices), whereas the KMZ algorithm is only practical for graphs with up to $\approx$ 300 vertices due to the bottleneck of solving the (enormous) LP.
\item The large clusters found by our algorithm can be meaningful, in that the algorithm partially discovers ``ground truth" clusters in real-world and synthetic instances. \\
\end{itemize}

Our experiments focus on the exact algorithm. We observe that empirically the exact algorithm is sufficiently fast.

\paragraph{Real dataset description.} 
We obtained datasets representing social networks from the Stanford Large Network Dataset Collection \cite{FBSocialCircles, feather, HepData}.\footnote{https://snap.stanford.edu/data/ego-Facebook.html} Specifically, we used the \textit{ego-Facebook} dataset containing 10 graphs that are subgraphs of a social network from Facebook. Each subgraph, or \textit{ego-network}, represents a specific user's friend list and the connections within it. We converted this to a complete, signed graph by representing a connection between users as a positive edge, and a non-connection as a negative edge. (See Tables \ref{table: fb_small_statistics} and \ref{table: fb_large_statistics} in Appendix \ref{appendix: add_plots_exp_sec} for statistics on these graphs.) Each ego-network is accompanied by ``ground truth" circles; each circle is a collection of vertices that the user has labelled as a community. We note that the circles are not necessarily partitions of the friend list, as they may overlap or not cover the entire friend list.  

For each Facebook ego-network, we applied our exact algorithm using the matrix multiplication implementation.  For five of the ego-networks that were of small enough size, we also solved the LP in order to bound the approximation ratio of our algorithm.  The LP solver used was Gurobi.  For the latter datasets, we also applied the KMZ algorithm as an additional means of comparison. Let $r_1$ be the radius in $L_t(\cdot)$ and let $r_2$ be the radius used to cut out $C_t$ in Algorithm \ref{KMZ-alg}. While Theorem \ref{thm: apx_thm} holds for $r_1 = 1/5$ and $r_2 = 2/5$, in practice these radii may give an objective value near the maximum positive degree in a sparse graph. We can obtain even better results than those guaranteed by Theorem \ref{thm: apx_thm} by setting the hyperparameters more conservatively. We did parameter sweeps (Figure \ref{fig: sweeps} in Appendix \ref{appendix: add_plots_exp_sec}), and found that $r_1 = r_2 = 0.7$ work well for our algorithm, and $r_1 = r_2 = 0.4$ work well for the KMZ algorithm, so we report the results using these parameters. See also Appendix \ref{appendix: double_parameter_sweep} for the best radii for each dataset. 
Finally, we applied the Pivot algorithm (described in Appendix \ref{appendix: pivot}) to all datasets for an additional comparison \cite{ACN-pivot}. See Tables \ref{table: objectives_fb_small} and \ref{table: objectives_fb_large} and Figure \ref{table: run-times_fb_small}. 

\paragraph{Quality of approximation.} For the five small datasets in Table \ref{table: objectives_fb_small}, the cost of our algorithm is at most 2 times the cost of the LP, and thus at most twice optimal. Our algorithm and the KMZ algorithm performed similarly in terms of objective value. In addition, the fractional cost of the correlation metric and the cost of the LP consistently differ by a factor of around 2. Finally, the objective value of our algorithm is typically slightly less than its fractional cost. For the large data sets (Table \ref{table: objectives_fb_large}) for which it was prohibitive to run the LP, we do not have a lower bound of optimal due to the LP not scaling, so we cannot bound the approximation ratio. For these datasets, we compare to Pivot, which our algorithm outperforms by a substantial margin.

\begin{center}
\begin{table}[H]
\small
\centering
\begin{tabular}{c | c | c | c | c | c } 
 \hline
  & fractional cost & LP objective & our objective  & KMZ objective  & Pivot objective  \\
\hline
FB 348 & 74.37 & 39.13 & 72 & 89 & 85.03 \\
\hline
FB 414 & 35.53 & 19.66 & 34 & 38 & 50.73 \\
\hline
FB 686 & 58.59 & 30.48 & 47 & 69 & 65.72 \\
\hline
FB 698 & 22.31 & 10.64 & 20 & 18 & 23.51\\
\hline
FB 3980 & 14.31 & 7.34 & 12 & 13 & 16.36 \\
\hline 
\end{tabular}
\captionsetup{width=.9\linewidth}
\caption{Comparison of the fractional values and objective values of our algorithm and the KMZ algorithm for the five small Facebook datasets. The column ``fractional cost" records the fractional cost of the correlation metric. Additionally, we run the Pivot algorithm; the recorded value is the objective value averaged over 500 trials.}
\label{table: objectives_fb_small}
\end{table}
\end{center}

\begin{center}
\begin{table}[H]
\small
\centering
\begin{tabular}{c | c | c | c | c | c } 
 \hline
  & fractional cost & our objective & Pivot objective  & our run-time & \#vertices \\
\hline
FB 0 & 64.02 & 49 & 71.78 & 0.20 & 333\\
\hline
FB 107 & 181.49 & 152 & 216.65 & 1.76 & 1034\\
\hline
FB 1684 & 103.99 & 93 & 130.71 & 0.98 & 786 \\
\hline
FB 1912 & 227.74 & 220 & 259.01 & 0.93 & 747 \\
\hline
FB 3437 & 98.36 & 107 & 99.1 & 0.47 & 534 \\
\hline 
\end{tabular}
\captionsetup{width=.9\linewidth}
\caption{Columns are as in Table \ref{table: objectives_fb_small}. The size of each dataset here makes the LP run-time prohibitive. Run-time is listed in seconds.}
\label{table: objectives_fb_large}
\end{table}
\end{center}

\paragraph{Run-time and scalability.} 
The run-time of our algorithm is significantly faster than that of the KMZ algorithm (Table \ref{table: run-times_fb_small}). For instance, on FB 348, which contains only 224 vertices, the KMZ algorithm took over 30 minutes, whereas our algorithm took a tenth of a second \footnote{Even if we terminate the LP early, e.g., when the primal and dual are within 1 of each other, this still takes at least 10 minutes. Note however that doing so would affect the distances outputted.}. In fact, we can quickly handle very large graphs; on a social network with 12,008 nodes \footnote{https://snap.stanford.edu/data/feather-lastfm-social.html}, our algorithm ran in just over 4 minutes. See Appendix \ref{appendix: scalability} for more details on scalability.

\paragraph{Comparison to ground truth clusters.} We also analyzed whether the clusters found by our algorithm discovered the ground truth circles identified by users. For each dataset, we considered ``large" clusters of size at least 10, since the small clusters (including several singleton clusters) are less meaningful. For each large cluster, we identified the ground truth circle containing the largest number of vertices from that cluster. The results are plotted in Figure \ref{fig: ground_truth_common_radius} in Appendix \ref{appendix: add_plots_exp_sec}.  We find that for datasets \textsc{FB} 348, 414, 686, 1684, and 1912, almost every large cluster is almost entirely contained in its best ground truth circle (i.e., between 80\% and 100\% of each large cluster is contained in its best ground truth circle). For other datasets, there is less evidence that ground truth circles are being discovered.

\begin{center}
\begin{table}[H]
\small
\centering
\begin{tabular}{c | c | c | c } 
 \hline
  & our run-time & KMZ run-time & \#vertices \\
\hline
FB 348 & 0.10 & 1847.99 & 224 \\
\hline
FB 414 & 0.06 & 207.92 & 150 \\
\hline
FB 686 & 0.06 & 337.9 & 168 \\
\hline
FB 698 & 0.02 & 3.42 & 61  \\
\hline
FB 3980 & 0.01 & 2.03 & 52 \\
\hline 
\end{tabular}
\captionsetup{width=.9\linewidth}
\caption{A comparison of the run-times (in seconds) of our algorithm and the KMZ algorithm on the five small Facebook datasets.} 
\label{table: run-times_fb_small}
\end{table}
\end{center}

\paragraph{Synthetic datasets.} 
We considered synthetic datasets for two reasons. The first is that running the LP and the KMZ algorithm are prohibitive on many real-world datasets. The second is to further test whether our algorithm discovers ground truth clusters. We took a graph with 100 vertices and 10 positive cliques of size 10 (so the graph admits a perfect clustering) and introduced 20 levels of noise. At each level $i$, we randomly flipped $45i$ edges to the opposite sign. We then applied our algorithm using $r_1 = r_2 = 0.7$ as before. We found that for up to 495 flips ($i=11$), the original clusters of size 10 were almost entirely preserved by our algorithm (in some cases, up to three vertices popped out into singleton clusters). We also found that for all levels of noise we considered, almost every cluster we found was at least 88\% contained in a ground truth cluster (only 6 clusters were an exception to this). See Appendix \ref{appendix: synthetic} for additional plots of these experiments. 

\section{Conclusion}
We presented a faster, completely combinatorial $O(1)$-approximation 
algorithm for Min Max correlation clustering. 
We constructed a fractional solution to our problem's LP based on the intersection sizes of the $+$ and $-$ neighborhoods of vertices, and then showed that the LP rounding algorithm by Kalhan, Makarychev, and Zhou (see \cite{KMZ19}) is successful with our hand-crafted fractional solution. 
By itself, this is a surprising result! 
It opens up the following question for future study: Given an LP, when can we use observable, combinatorial properties of the underlying instance's structure to form a provably good fractional solution to the LP?
This general framework could lead to big run-time improvements for other problems.

Another future direction is to search for a hand-crafted fractional solution for the $\ell_1$ norm to obtain other combinatorial algorithms for classic correlation clustering. Moreover, while we show the success of the \dn~for the $\ell_{\infty}$
norm, it is possible that it works as a surrogate fractional solution for other $\ell_p$ norms too.
While these directions are theoretically interesting in their own right, there is practical motivation for finding fast algorithms for other $\ell_p$ norms, since $p \in (1, \infty)$ interpolates between the competing objectives of local fairness ($p = \infty$) and global optimality ($p = 1$).

\printbibliography

\appendix

\section{Rounding Algorithm and Run-time}
\label{sec: lp_rounding_alg}

We now detail the rounding algorithm of Kalhan, Makarychev, and Zhou
that we will leverage \cite{KMZ19}.  
For vertex $u \in V$ let the ball of radius $\rho$ 
around $u$ with respect to a semi-metric $x$ on $V$
be defined as $\textsf{Ball}(u,\rho) = \{v \in V \mid x_{uv} \leq \rho\}$.
The algorithm is iterative, where vertices are clustered in each iteration.
At step $t$, the set of unclustered vertices is denoted $V_t \subseteq V$.
From $V_t$, a special vertex is chosen to be the cluster center, 
specifically the vertex that maximizes 
$$
L_t(u)= \sum_{v \in \textsf{Ball}(u,r) \cap V_t} {r-x_{uv}},
$$
which indicates how packed towards the center the vertices in
$\textsf{Ball}(u,r) \cap V_t$ are.\\

\begin{algorithm} \label{KMZ-alg}[Rounding algorithm]
\end{algorithm}

\begin{minipage}{13cm}
\rule{12cm}{0.4pt}\\
\noindent \textbf{Input: }Semi-metric $x$ on $V$.

\noindent \textbf{Output: }Clustering $\mathcal{C}$.
\begin{enumerate}
    \item Let $V_0 = V$, $r = 1/5$, $t=0$.
    \item \textbf{while} ($V_t \neq \emptyset$)
    \begin{itemize}
        \item Find $u^*_t = \arg \max_{u \in V_t} L_t(u) =  \arg \max_{u \in V_t} \sum_{v \in \text{Ball}(u,r) \cap V_t} {r-x_{uv}}$. 
        \item Create ${C_t = \text{Ball}(u^*_t,2r) \cap V_t}$.
        \item Set $V_{t+1} = V_t \setminus C_t$ and $t=t+1$.
    \end{itemize}
    \item Return ${\mathcal{C} = (C_0, \dots, C_{t-1})}$.
\end{enumerate}
    \rule{12cm}{0.4pt}
\end{minipage}\\

Let $\textsf{LP}(u,v)$ be the cost of edge $(u,v)$ to the LP in its objective value, 
so one can set $x_{uv}=\textsf{LP}(u,v)$ if $(u,v) \in E^+$ and 
$x_{uv}=1-\textsf{LP}(u,v)$ if $(u,v) \in E^-$. 
For the output of the algorithm above, 
let $\textsf{ALG}(u,v) = \mathds{1}((u,v) \text { is in disagreement})$ 
and $\textsf{ALG}(u) = \sum_{v \in V} \textsf{ALG}(u,v)$.
Kalhan, Makarychev, and Zhou show that 
\begin{equation}\label{eq: KMZ}
    \textsf{ALG}(u) = \sum_{v \in V}\textsf{ALG}(u,v) \underset{*}{\leq} 
5 \cdot \sum_{v \in V}\textsf{LP}(u,v) \leq 5 y(u),
\end{equation} 
which leads to a 5 approximation algorithm for any $\ell_p$ norm 
for complete graphs. 
The technical work in their result is in showing the inequality
marked with * in Equation \ref{eq: KMZ}. (The second inequality is due to the LP being a relaxation.)

\paragraph{Run-time}
We will first justify that the run-time of the rounding algorithm is $O(n^2)$, 
and can be obtained through the following procedure:
\begin{itemize}
    \item For each $u \in V$, precompute $L_t(u) = \sum_{v \in Ball(u,r)} r - d_{uv}$. This takes $O(n^2)$ time. 
    \item At iteration $t$, it takes $O(n)$ time to find the max of $L_t(u)$ over all unclustered vertices, i.e., over $u \in V_t$, and $O(n)$ time to create the cluster. 
    This contributes $O(n^2)$ time over all iterations. 
    \item When a vertex is removed, its contribution to $L_t(u)$ for remaining vertices $u$ must be removed as well. This takes $O(n)$ time for each vertex that is removed, as it may be a member of $O(n)$ balls. Further, each vertex is only removed once. So the updates to the $L_t$ values take $O(n^2)$ time overall.
\end{itemize}

In all, the run-time of the full KMZ algorithm is dominated by the time it takes to solve the LP to obtain the semi-metric $x$.  
The LP contains $O(n^2)$ many variables and $O(n^3)$ many constraints, 
for $n = |V|$. 
Given the current best algorithms for solving linear programs,
this obliviously gives a run-time that is no better than $O(n^{2 \omega})$ for solving the LP \cite{cohen2021solving}. 
We justify that this is a reasonable theoretical benchmark to compare against for run-time.
On sparse networks, i.e. when the graph on the $+$ edges 
is sparse, one might wonder whether it is possible to reduce the number of variables that the solver must solve for by fixing the values of the LP variables for pairs of vertices whose positive neighborhoods have empty intersection.
However, it is easy to come up with small examples that show
one cannot do this without violating the triangle inequality, 
and is is totally unclear if one could even guarantee some approximate triangle inequality here. 
Lastly, it might be possible to more quickly obtain an LP solution which is provably approximately optimal by solving the LP, 
but such study would be quite technical and a completely different approach.

\section{Choices of Constants in Section \ref{sec: approx_tri_suff}} \label{appendix: approx_tri_requirements}
Let $\eps > 0$ be sufficiently small, 
$\delta_1 = 3 + h_4(\eps)$, and $\delta_2 = h_4(\eps)$, where $h_4(\eps)$ is as in Proposition \ref{prop: pseudo_triangle}. Note that $h_4(\eps) \rightarrow 0$ as $\eps \rightarrow 0$. Take 
\begin{align*}
r &= \frac{1-\delta_2 - \delta_1 \delta_2 - \delta_1^3 \delta_2 - \delta_1^2 \delta_2}{\delta_1^2 + \delta_1^3(\delta_1+1) + \delta_1 + 1} \\
c_1 &= \delta_1 + \frac{\delta_2}{r} \\
b &= (c_1 + 1)\delta_1 + \frac{\delta_2}{r} \\
c_2 &= \delta_1(b+1) + \frac{\delta_2}{r}.
\end{align*}

One can verify that the following inequalities, as needed in Section \ref{sec: approx_tri_suff}, hold. 

\begin{align*}
    &b \geq 1 \qquad &&c_2 \cdot r<1 \\
    &c_1 \leq b < c_2 \qquad &&1-2 \delta_1 b\cdot r - \delta_2-r\geq 0\\
    &\frac{1}{\delta_1} b \cdot r - c_1 \cdot r -r - \delta_2/\delta_1\geq 0
    \qquad &&\frac{1}{\delta_1 }c_2 \cdot r - b \cdot r - r - \delta_2/\delta_1 \geq 0\\
    &\frac{1}{\delta_1}c_1 \cdot r- \delta_2/\delta_1 \geq r
    \qquad 
    &&1 - (\delta_1 b +\delta_1) r - \delta_2 \geq  r\\
    &(\frac{1}{\delta_1} b-1)r - \delta_2/\delta_1 \geq r  \qquad
    && 1 - (\delta_1 c_2 + \delta_1)r - \delta_2\geq r 
\end{align*}


\section{Supplementary Experiment Information} \label{appendix: exp}

\subsection{Description of Pivot algorithm} \label{appendix: pivot}
The Pivot algorithm of  Ailon, Charikar, and Newman is a randomized algorithm that gives a 3-approximation in expectation for classic correlation clustering (i.e., $\ell_1$ norm) on a complete graph \cite{ACN-pivot}. Pivot may perform poorly for the Min Max objective \cite{PM16}. The Pivot algorithm is as follows. Sample a uniformly random ordering of the vertices. Visit vertices in this order. Upon visiting a vertex, check whether it has been marked as clustered. If it has, visit the next vertex. If not, call the current vertex a pivot, and open a new cluster consisting of the pivot and all unclustered vertices that are in the positive neighborhood of the pivot. Mark these vertices as clustered. 

\subsection{Additional plots for Section \ref{sec: experiments}} \label{appendix: add_plots_exp_sec}
This section contains additional plots for the experiments discussed in Section \ref{sec: experiments}. 

\begin{center}
\begin{table}[H]
\small
\centering
\begin{tabular}{c | c | c | c | c | c } 
 \hline
  & \textsc{FB 348} & \textsc{FB 414} & \textsc{FB 686} & \textsc{FB 698} & \textsc{FB 3980}\\ [0.5ex] 
 \hline
\#vertices & 224 & 150 & 168 & 61 & 52 \\ [1ex] 
 \hline
 \#edges & 6384 & 3386 & 3312 & 540 & 292 \\ 
 \hline
 max positive degree & 100 & 58 & 78 & 30 & 19 \\
 \hline
 \end{tabular}
\captionsetup{width=.9\linewidth}
\caption{Graph statistics for the five small Facebook datasets in Table \ref{table: objectives_fb_small}.}
\label{table: fb_small_statistics}
\end{table}
\end{center}

\begin{center}
\begin{table}[H]
\small
\centering
\begin{tabular}{c | c | c | c | c | c } 
 \hline
  & \textsc{FB 0} & \textsc{FB 107} & \textsc{FB 1684} & \textsc{FB 1912} & \textsc{FB 3437}\\ [0.5ex] 
 \hline
\#vertices & 333 & 1034 & 786 & 747	& 534 \\ [1ex] 
 \hline
 \#edges & 5038 & 53498 & 28048 & 60050 & 9626 \\ 
 \hline
 max positive degree & 78 & 254 & 137 & 294 & 108 \\
 \hline
  \end{tabular}
\captionsetup{width=.9\linewidth}
\caption{Graph statistics for the five large Facebook datasets in Table \ref{table: objectives_fb_large}.}
\label{table: fb_large_statistics}
\end{table}
\end{center}

\begin{figure}[H]
\centering
\begin{minipage}{.5\textwidth}
  \centering
  \includegraphics[width=.9\linewidth]{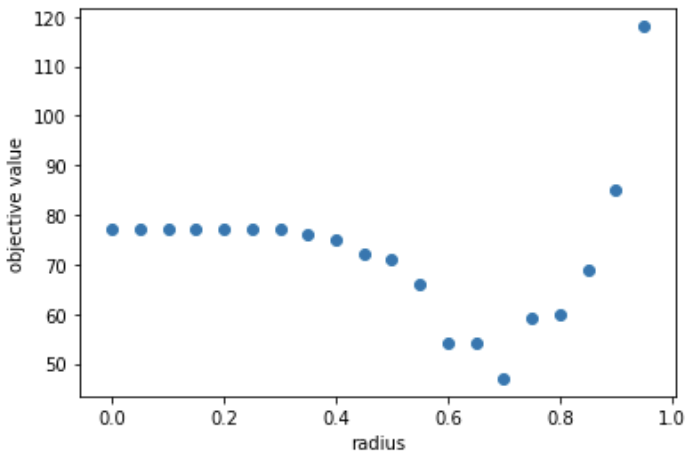}
\end{minipage}%
\begin{minipage}{.5\textwidth}
  \centering
  \includegraphics[width=.9\linewidth]{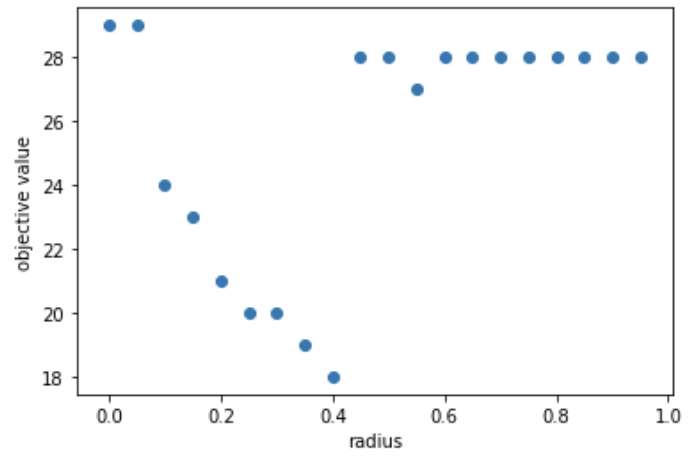}
\end{minipage}
\captionsetup{width=.9\linewidth}
\caption{Plots showing how the objective value changes for our algorithm (left) and the KMZ algorithm (right) as we sweep over a common radius $r_1 = r_2$ for datasets FB 686 (left) and FB 698 (right). Plots for other datasets are similar, so we used radii of $0.7$ and $0.4$ for our algorithm and the KMZ algorithm, respectively, in Tables \ref{table: objectives_fb_small} and \ref{table: objectives_fb_large}.}
\label{fig: sweeps}
\end{figure}

\begin{figure}[H]
\centering
\includegraphics[width=0.62\textwidth]{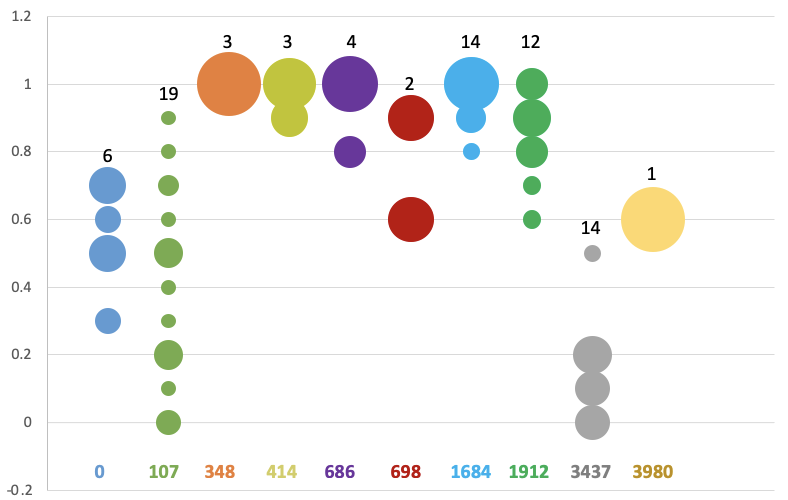}
\captionsetup{width=.9\linewidth}
\caption{For each dataset (bottom), the bubbles above it quantify the extent to which large clusters (size $\geq 10$) found by our algorithm are contained in ground truth clusters. The number of large clusters for each dataset is at the top of the column. Each bubble represents a certain number of large clusters for the corresponding dataset; the size of the bubble is proportional to how many large clusters it represents. For each large cluster, we found its best ground truth cluster, i.e., the one containing the largest number of vertices from that cluster. The $y$-axis represents the proportion of a large cluster contained in its best ground truth cluster. For example, for dataset \textsc{FB 698}, consider a given large cluster that has between 90\% (inclusive) and 100\% (exclusive) of its vertices contained in its best ground truth cluster. It is accounted for by the red circle vertically positioned at 0.9. A large cluster that has between 60\% and 70\% of its vertices contained in its best ground truth cluster, is accounted for by the red circle vertically positioned at 0.6. Since the sizes of these two red circles are equal, half of the large clusters fall into each of these two categories. Notice that for datasets \textsc{FB 348, 414, 686 1684,} and 1912, there are bubbles located at the $y$-axis position of 1; these bubbles represent large clusters that are \textit{completely} (i.e., 100\%) contained in a ground truth cluster. For dataset FB 348 in particular, \textit{every} large cluster is 100\% contained in a ground truth cluster.}
\label{fig: ground_truth_common_radius}
\end{figure}

\subsection{Results for double parameter sweep on Facebook datasets} \label{appendix: double_parameter_sweep}
For ease of discussion in Section \ref{sec: experiments} and Appendix \ref{appendix: add_plots_exp_sec}, we applied our algorithm and the KMZ algorithm to the Facebook datasets by enforcing a common radius $r = r_1 = r_2$. We used $r=0.7$ for our algorithm and $r = 0.4$ for the KMZ algorithm on \textit{all datasets}. In this section, we present a more tailored analysis, where for \textit{each dataset} and \textit{each algorithm}, we find the best $r_1$ and $r_2$ for that dataset (without requiring that $r_1 = r_2$). The results are reported in Tables \ref{table: fb_small_diff_radii} and \ref{table: fb_large_diff_radii}. While the results are similar to those reported in Section \ref{sec: experiments} and Appendix \ref{appendix: add_plots_exp_sec}, we include these results for completeness.

\begin{center}
\begin{table}[H]
\small
\centering
\begin{tabular}{c | c | c | c | c | c } 
 \hline
  & \textsc{FB 348} & \textsc{FB 414} & \textsc{FB 686} & \textsc{FB 698} & \textsc{FB 3980}\\ [0.5ex] 
 \hline
 fractional cost & 74.37	& 35.53	& 58.59	& 22.31	& 14.31 \\
  \hline
 LP objective & 39.13 & 19.66 & 30.48	& 10.64 & 7.34 \\ [1ex]
 \hline
 our objective & 71 & 31 & 43 & 18 & 12 \\
 \hline
KMZ objective  & 69 & 28 & 47 & 17 & 13 \\ [1ex]
 \hline
Pivot objective  & 85.03 & 50.73 & 65.72 & 23.51 & 16.36 \\ [1ex]
 \hline
   ($r_1$, $r_2$) & (0.45, 0.7) & (0.1, 0.65) & (0.3, 0.75) & (0.8, 0.8) & (0.7, 0.7) \\ [1ex] 
 \hline
KMZ $(r_1, r_2)$ &	(0.2, 0.45) &  (0.05, 0.6) & (0.3, 0.6) & (0.05, 0.5) & (0.3, 0.6) \\ [1ex]
 \hline
\end{tabular}
\captionsetup{width=.9\linewidth}
\caption{Experimental results for the five small Facebook datasets. The pairs $(r_1, r_2)$ and KMZ $(r_1, r_2)$ refer to an optimal pair of radii in a double parameter sweep (instead of setting $r_1 = r_2$ as in Section \ref{sec: approx_tri_suff} and Appendix \ref{appendix: add_plots_exp_sec}), which were used in the rounding phase of each algorithm.}
\label{table: fb_small_diff_radii}
\end{table}
\end{center}

\begin{center}
\begin{table}[H]
\small
\centering
\begin{tabular}{c | c | c | c | c | c } 
 \hline
  & \textsc{FB 0} & \textsc{FB 107} & \textsc{FB 1684} & \textsc{FB 1912} & \textsc{FB 3437}\\ [0.5ex] 
 \hline
 fractional cost & 64.02 & 181.49 & 103.99 & 227.74	& 98.36 \\
 \hline
our objective & 49 & 134 & 93	& 187 & 77 \\
 \hline
 Pivot objective & 71.78 &	216.65 & 130.71	& 259.01 & 99.1 \\ [1ex]
 \hline
 ($r_1$, $r_2$) & (0.7, 0.7) & (0.65, 0.65) & (0.7, 0.7) & (0.65, 0.65) & (0.3, 0.7) \\ [1ex] 
 \hline
\end{tabular}
\captionsetup{width=.9\linewidth}
\caption{Experimental results for the five larger Facebook datasets. As in Table \ref{table: fb_small_diff_radii}, $(r_1, r_2)$ refers to the optimal radii in a double parameter sweep.}
\label{table: fb_large_diff_radii}
\end{table}
\end{center}

\subsection{Synthetic data: perfect clusterings with noise} \label{appendix: synthetic}

Below we show various plots for the synthetic experiments discussed in Section \ref{sec: experiments}.

\begin{figure}[H]
\centering
\begin{minipage}{.5\textwidth}
  \centering
  \includegraphics[width=.9\linewidth]{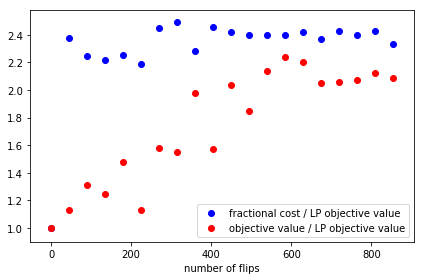}
\end{minipage}%
\begin{minipage}{.5\textwidth}
  \centering
  \includegraphics[width=.9\linewidth]{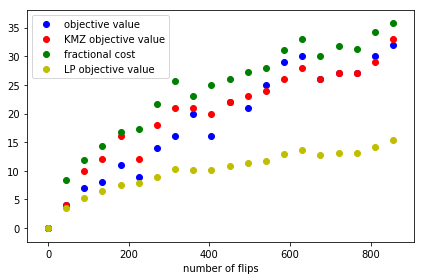}
\end{minipage}
\captionsetup{width=.9\linewidth}
\caption{In the above plots, ``fractional cost" refers to the fractional cost of our correlation metric. The term ``objective value" refers to that of our algorithm, while ``KMZ objective value" refers to that of the KMZ algorithm. These objective values are with respect to $r_1=r_2=0.7$ for our algorithm and $r_1 = r_2 = 0.4$ for the KMZ algorithm.}
\label{fig: synth_scatter}
\end{figure}

\begin{figure}[H]
\centering
\includegraphics[width=0.5\textwidth]{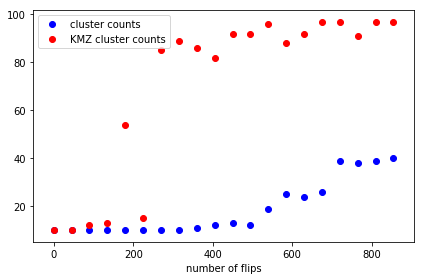}
\captionsetup{width=.9\linewidth}
\caption{The term ``cluster counts" refers to the number of clusters output by our algorithm, and ``KMZ cluster counts" for the number output by the KMZ algorithm.}
\label{fig: cluster_cts_minus_pivot}
\end{figure}

\subsection{Scalability of our algorithm} \label{appendix: scalability}
We ran our exact algorithm (using the matrix multiplication implementation) on three large datasets with approximately 10,000 vertices each to show that our algorithm scales. The first dataset, denoted LastFM \footnote{https://snap.stanford.edu/data/feather-lastfm-social.html}, is a social network of users of the music service LastFM in Asia, where vertices represent users and positive edges represent mutual follower relationship \cite{feather}. The other two datasets, denoted ca-HepTh and ca-HepPh \footnote{https://snap.stanford.edu/data/ca-HepTh.html, https://snap.stanford.edu/data/ca-HepPh.html},  are collaboration networks of high energy physics authors on arXiv, where vertices represent authors and positive edges represent co-authors \cite{HepData}. See Table \ref{table: run-times_scalable} for our algorithm's run-time on these three datasets. We observe that the algorithm takes approximately two to four minutes on each dataset.  

\begin{center}
\begin{table}[H]
\small
\centering
\begin{tabular}{c | c | c | c } 
 \hline
  & our run-time & \#vertices & \#edges \\
\hline
LastFM & 102.74 & 7624 & 27806 \\
\hline
ca-HepTh & 165.42 & 9877 & 51971 \\
\hline
ca-HepPh & 250.91 & 12008 & 237010 \\
\hline 
\end{tabular}
\caption{Run-times (in seconds) of our exact algorithm on three large datasets, along with their sizes.} 
\label{table: run-times_scalable}
\end{table}
\end{center}

\end{document}